%% file: paper.tex
\documentclass[11pt]{article}
\pdfoutput=1

\usepackage[utf8]{inputenc}
\usepackage[T1]{fontenc}
\usepackage[american]{babel}

\usepackage[margin=1in]{geometry}

\usepackage{amsmath}
\usepackage{amssymb}
\usepackage{complexity}

\usepackage[boxruled,linesnumbered]{algorithm2e}
\SetKw{Accept}{accept}
\SetKw{Reject}{reject}

\usepackage{csquotes}
\usepackage{enumerate}

\usepackage{hyperref}

\usepackage{tikz}
\usepackage[sort]{standalone}

\usepackage{subcaption}

\usepackage[
  backend=biber,bibencoding=utf8,sortcites=true,maxbibnames=99
]{biblatex}
\AtEveryBibitem{%
  \clearfield{isbn}
  \clearfield{issn}
  \clearfield{note}
  \clearfield{url}
}

\let\H\relax

\usepackage{ammacros}
\usepackage{amthm}

\DeclareMathOperator{\bn}{bn}
\DeclareMathOperator{\viol}{viol}
\DeclareMathOperator{\wn}{wn}

\newcommand{\ENV}{\mathtt{ENV}}

\newlang{\oBP}{1\mhyph BP}
\newlang{\tauBP}{\tau\mhyph BP}

\usepackage[capitalize]{cleveref}
\crefname{claim}{Claim}{Claims}

\title{Testing Spreading Behavior in Networks with Arbitrary Topologies}
\author{Augusto Modanese\\
  Aalto University, Finland\\
  \texttt{augusto.modanese@aalto.fi}
  \and
  Yuichi Yoshida\\
  National Institute of Informatics\\
  \texttt{yyoshida@nii.ac.jp}}
\date{}

\begin{document}

\maketitle

\begin{abstract}
  Given the full topology of a network, how hard is it to test if it is evolving
  according to a local rule or is far from doing so?
  Inspired by the works of Goldreich and Ron (J.~ACM, 2017) and Nakar and Ron
  (ICALP, 2021), we initiate the study of property testing in dynamic
  environments with arbitrary topologies.
  Our focus is on the simplest non-trivial rule that can be tested, which
  corresponds to the $\oBP$ rule of bootstrap percolation and models a simple
  spreading behavior:
  Every \enquote{infected} node stays infected forever, and each
  \enquote{healthy} node becomes infected if and only if it has at least one
  infected neighbor.
  Our results are subdivided into two main groups:
  \begin{itemize}
    \item If we are testing a single time step of evolution, then the query
    complexity is $O(\Delta/\eps)$ or $\tilde{O}(\sqrt{n}/\eps)$ (whichever is
    smaller), where $\Delta$ and $n$ are the maximum degree of a node and the
    number of vertices in the underlying graph, respectively.
    We also give lower bounds for both one- and two-sided error testers that
    match our upper bounds up to $\Delta = o(\sqrt{n})$ and $\Delta =
    O(n^{1/3})$, respectively.
    If $\eps$ is constant, then the first of these also holds against adaptive
    testers.
    \item For the setting of testing the environment over $T$ time steps, we
    give two algorithms that need $O(\Delta^{T-1}/\eps T)$ and
    $\tilde{O}(\abs{E}/\eps T)$ queries, respectively, where $E$ is the set of
    edges of the underlying graph.
  \end{itemize}
  All of our algorithms are one-sided error, and all of them are also
  non-adaptive, with the single exception of the more complex
  $\tilde{O}(\sqrt{n}/\eps)$-query tester for the case $T = 2$.
\end{abstract}


\section{Introduction}

Imagine we are observing the state of a network as it evolves over time.
The network is static and we have complete knowledge about the connections;
it is too large for us to keep track of the state of every single node, though
nevertheless we are able to query nodes directly and learn their states.
We might hypothesize that the global behavior can be explained by a certain
local rule that is applied at every node, and we would like to verify if our
hypothesis is correct or not.
In this paper, we focus on this question:
\emph{How hard is it to test, given a local rule $R$, if the network is
following $R$ or is far from doing so?}

Following previous works \cite{goldreich17_learning_jacm,nakar21_back_icalp}, we
refer to the series of configurations assumed by the network over time as the
\emph{environment} $\ENV$ that we are observing.
The network itself is static and its connections defined by a graph $G = (V,E)$.
The local rule $R$ is a map (admitting a finite description) from the states
that a node observes in its neighborhood (including the node itself) to the new state it will assume in the
next time step.
Plausible scenarios that could be modeled in this context include not only
rumor dissemination in social networks but also spreading of infectious
diseases (where the connections between nodes represent proximity or contact
between the organisms that we are observing).
As is common in property testing \cite{bhattacharyya22_property_book}, we assume
that the bottleneck of this problem is keeping track of the states across the
entire network, and hence we consider only the number of queries made by a
testing algorithm as its measure of efficiency (and otherwise assume that the
algorithm has access to unbounded computational resources).

\subsection{Problem Setting}

There exist two previous works
\cite{goldreich17_learning_jacm,nakar21_back_icalp} that study the problem of
determining whether $\ENV$ evolves according to $R$ in the context of property
testing (and, in the case of \cite{goldreich17_learning_jacm}, also in the
context of learning theory).
In these works, the structure underlying the environment is always a cellular
automaton (in the case of \cite{nakar21_back_icalp} one-dimensional, whereas 
\cite{goldreich17_learning_jacm} also considers automata of multiple
dimensions), and thus $\ENV$ corresponds to the time-space diagram of such an
automaton.
This perspective is certainly meaningful when we are interested in phenomena
that take place on a lattice or can be adequately represented in such grid-like
structures, for instance the movement of particles on a surface or across
three-dimensional space.
Nevertheless there are limits as to what can be modeled in this way.
A prominent example are social networks, in which the connections hardly fit
well into a regular lattice (even with several dimensions).

In this work we cast off these restraints and instead take the radically
different approach of making no assumptions about the underlying structure or
the space it is embedded in.
Our only requirement is that it corresponds to a static graph $G$ that is known
to us in advance.
This leaves a much broader avenue open when it comes to applications.
In addition, the rule that we consider is effectively the simplest rule possible
in such a setting that is not trivial.
As we will see, despite the rule being very simple, it is rather challenging to
fully determine the complexity of the problem.
Indeed, compared to the previous works mentioned above, it might seem as if our
progress is more modest; however, one should keep in mind that, in our case, the
underlying network $G$ has a much more rich structure (whereas in cellular
automata we are dealing with a highly regular one).

The rule that we study is the $\oBP$ rule of \emph{bootstrap percolation}
\cite{gregorio09_bootstrap_ecss,zehmakan19_tight_lata,janson12_bootstrap_aap,%
balogh07_bootstrap_rsa}.
For $\tau \in \N_0$, the rule $\tauBP$ is defined based on two states,
\emph{black} and \emph{white}, as follows:
If a node is black, then it always remains black; if a node is white, then it
turns black if and only if it has at least $\tau$ black neighbors.
These rules were originally inspired in the behavior observed in certain
materials, and they are very naturally suited for modeling spreading phenomena.

As seen from the lenses of property testing, testing for the $\oBP$ rule in some
sense resembles the setting of monotonicity testing
\cite{fischer02_monotonicity_stoc}.
Though we cannot directly apply one strategy to the other, if we view black as
$1$ and white as $0$, then in both cases we have a violation whenever we see a
$1$ preceding a $0$.
The difference is that in $\oBP$ \emph{every} $1$ must arise from a preceding
$1$, whereas in the case of monotonicity we are happy if an isolated $0$
spontaneously turns into a $1$.

Another way of modeling the $\oBP$ rule is as a \emph{constraint satisfaction
problem} (CSP).
CSPs have been studied in the context of property testing to some extent 
\cite{bhattacharyya13_algebraic_icalp,chen19_constant_siamjc}.
We can characterize $\oBP$ by two constraints:
A black node in step $t$ implies every one of its neighbors is also black in
step $t+1$; meanwhile, a node is white in step $t+1$ if and only if every one of
its neighbors in step $t$ was white.
Then we can recast testing if $\ENV$ follows $\oBP$ as testing if $\ENV$ is a
satisfying assignment for these constraints.
Nevertheless, although this seems to be a useful rephrasing of the problem, the
current methods in CSPs in the context of property testing are not sufficient to
tackle it.
And, even if we could indeed test either constraint with a sublinear number of
queries, $\ENV$ being close to satisfying both constraints would not necessarily
imply that $\ENV$ is close to satisfying their intersection.

\subsection{Results and Techniques}

We now present our results and the methods used to obtain them.
As this is a high-level discussion, formal definitions are postponed to
\cref{sec:preliminaries}, which the reader is invited to consult as needed.

The relevant parameters for the results are the number of nodes $n$ in the graph
$G = (V,E)$, the number of steps $T$ during which the environment $\ENV$
evolves, the maximum degree $\Delta$ of $G$, and the accuracy parameter $\eps >
0$.
The size of the environment is $nT$, which is the baseline for linear complexity
in this context (instead of $n$).
We write $\ENV \in \oBP$ to indicate that $\ENV$ follows the $\oBP$ rule and
$\dist(\ENV,\oBP) \ge \eps$ when it is $\eps$-far from doing so, that is, one
must flip at least $\eps n T$ colors in $\ENV$ in order for $\oBP$ to be obeyed
everywhere.
(As already mentioned, see \cref{sec:preliminaries} for the precise
definitions.)

In most cases we will be interested in optimizing the dependency of the query
complexity on $\Delta$.
This is due to the fact that, intuitively, graphs with small $\Delta$ should be
easier to verify locally, that is, by looking only at each node's neighborhood.
(Indeed, this is the strategy followed by the first algorithm we present below
in \cref{thm:alg-low-degree}.)

Another desirable property that we wish our algorithms to have is
\emph{non-adaptiveness}; that is, the algorithm first produces a list of queries
(without looking at the input), gathers their results, and then decides whether
to accept or not.
This is in contrast to an \emph{adaptive} algorithm, which may perform later
queries based on the answers it has seen so far.
A third property \enquote{in-between} these is \emph{time-conformability},
meaning that the algorithm does not make queries in step $t$ if it has already
queried nodes at some later step $t' > t$.
It is easy to see that the existence of a non-adaptive algorithm with query
complexity $q$ implies a time-conforming algorithm with the same complexity:
Just gather the $q$ queries in a list, sort them according to the time step
queried, and then execute the queries in order.
The converse is not true in general, however, since a time-conforming algorithm
might choose its queries in a later time step based on what it has seen
beforehand (or in the same time step, even).

Recalling our motivation of testing the evolution of huge networks, we see that
non-adaptive algorithms are the most desirable because the queries may all be
performed in parallel (at each time step).
In case this cannot be achieved, an adaptive, time-conforming algorithm is
still satisfactory, even though it might require \enquote{freezing} the network
at a specific time step (so that the algorithm has time to gather the results
received and decide on the next queries to make on the same time step).
Adaptive algorithms that violate time-conformability are not particularly
desirable since they require \enquote{rewinding} the state of the network back
in time.
Nevertheless, depending on the nodes' capabilities, we might still be able to
find strategies to cope with this; for example, if $T$ is small, it is plausible
to require nodes to cache their state in previous steps (and thus they can
answer any of the algorithm's queries, even about previous states).

In this paper, we study two different settings: testing a single time step of
evolution ($T = 2$) and testing multiple steps ($T > 2$).
In the first case we prove both upper and lower bounds, which also match up to
certain values of $\Delta$.
In the second we show only upper bounds, but which suffice to demonstrate that
the problem admits non-trivial testers, at least for moderate (non-constant)
values of $\Delta$.

\subsubsection{The Case \texorpdfstring{$T = 2$}{T = 2}}
\label{sec:intro-results-t-2}

Let us first discuss our results for the case where $T = 2$.
In this case there is a natural graph-theoretical rephrasing of the problem:
For $t \in \{1,2\}$, let $S_t$ be the set corresponding to $\ENV(\cdot,t)$ where
we see $\ENV(\cdot,t)$ as an indicator function (i.e,. $S_t$ is exactly the set
of nodes $v \in V$ for which $\ENV(v,t) = 1$).
Then $\ENV \in \oBP$ if and only if $S_2$ is dominated by $S_1$ (in
graph-theoretic terms).\footnote{%
  Technically the definition of domination is so that $A$ dominates $B$ if and
  only if $B \subseteq A \cup N(A)$.
  Using this definition the equivalence is only true if the graph $G$ contains
  self-loops everywhere.
  (Nevertheless, adding self-loops everywhere does not impact the maximum
  degree, which is the relevant parameter here.)
  The equivalence is certainly true if we change the definition so that $A$
  dominates $B$ if $B \subseteq N(A)$.
}
From this perspective, the distance from $\ENV$ to $\oBP$ is the (relative)
total number of nodes we need to add or remove from either of $S_1$ or $S_2$ in
order for the domination property to be satisfied.

The hardness of the problem in this case is highly dependent on the maximum
degree $\Delta$.
Our first result is that there is a very natural and simple algorithm that
achieves query complexity $O(\Delta/\eps)$.

\begin{restatable}{theorem}{restateThmAlgLowDegree}%
  \label{thm:alg-low-degree}
  Let $T = 2$ and $\eps > 0$. 
  There is a \emph{non-adaptive, one-sided error} algorithm with query
  complexity $O(\Delta / \eps)$ that decides whether $\ENV \in \oBP$ or
  $\dist(\ENV,\oBP) \ge \eps$.
\end{restatable}

The algorithm simply selects nodes at random and then queries their entire
neighborhoods in both time steps.
Since we are dealing with a local rule, this is sufficient to detect if $\ENV$
contains too many violations of the rule or not.
One detail that needs care here is that, in general, our notion of distance does
\emph{not} match the number of violations of the rule.
Nevertheless, as we show, the cases where it does not are only playing in our
favor, and so this strategy always succeeds.

It turns out that this algorithm is optimal when we are in regimes where there
is a constant $b \ge 2$ such that $\eps = \Omega(\Delta^b/n)$.
We also prove lower bounds for the case where $b \ge 1$, which are especially
useful in regimes where $\Delta$ is larger than $\sqrt{n}$.

\begin{restatable}{theorem}{restateThmLBOneSided}%
  \label{thm:lb-one-sided}
  There is a constant $\eps_0 > 0$ such that the following holds:
  Let $\eps = \Omega(\Delta^b/n)$ be given where $b \ge 1$ is constant, and let
  $\eps \le \eps_0$.
  Then deciding if $\ENV \in \oBP$ or is $\eps$-far from $\oBP$ with a
  \emph{one-sided error} tester requires at least $q$ queries in general, where:
  \begin{enumerate}
    \item If $b > 2$, then $q = \Omega(\Delta/\eps)$ if the tester is
    non-adaptive or $q = \Omega(1/\eps + \Delta)$ if it is adaptive.
    \item If $b = 2$, then $q = \Omega(\Delta/\eps\log\Delta)$ if the tester is
    non-adaptive or $q = \Omega(1/\eps + \Delta/\log\Delta)$ if it is adaptive.
    \item If $1 \le b < 2$, then $q = \Omega(\Delta^{b-1}/\eps)$ if the tester
    is non-adaptive or $q = \Omega(1/\eps + \Delta^{b-1})$ if it is adaptive.
  \end{enumerate}
  Note the lower bounds hold even for adaptive testers in general, \emph{even
  for those that do not respect time-conformability}.
\end{restatable}

If we are only interested in the regime where $\eps$ is constant, then setting
$b = \log_\Delta n$ above we obtain a lower bound of $\Omega(\Delta)$ whenever
$\Delta = O(n^{1/2-c})$ for a constant $c > 0$.
This is matched by the upper bound of \cref{thm:alg-low-degree}.
For $\Delta = \Theta(\sqrt{n})$, the lower bound is $\Omega(\Delta/\log n)$; for
larger $\Delta$ the lower bound becomes $\Omega(n/\Delta)$ and thus deteriorates
as $\Delta$ increases.

The lower bound is based on an adequate construction of expander graphs.
More specifically, the expanders we construct are bipartite, $\Delta$-regular,
and have \emph{distinct} expansion guarantees for sets of nodes on either side.
This is needed because in one direction the expansion is giving the
$\eps$-farness of the instances we create; meanwhile expansion in the other
direction yields the actual lower bound on the number of queries that a correct
algorithm must make.

The hard instances themselves are simple:
We color a moderately large set $B$ of randomly chosen nodes black in the second
time step and leave the rest colored white.
The intuition is that, since $B$ is chosen at random, it will not match nicely
with a cover $C = \bigcup_{u \in S} N(u)$ induced by some set of nodes $S$ in
the first time step; that is, the symmetric difference between $B$ and $C$ will
likely be large, giving us $\eps$-farness.
At the same time, since a one-sided error algorithm $A$ cannot reject good
instances, it is hard for it to detect that there is something wrong with $B$
without having to \enquote{cover} a considerable number of nodes in either step.
Indeed, in order to ascertain that a node $v \in B$ is incorrect, $A$ must
verify that there is no black node in $N(v)$ in the first step; if the existence
of some black $u \in N(v)$ is compatible with its view, then there is no
contradiction to $v$ being black, and hence $A$ cannot reject.
Querying all of $N(v)$ requires $\Omega(\Delta)$ queries, but it is also
possible for $A$ to determine the colors \emph{indirectly} by querying neighbors
of nodes in $N(v)$ (since $u \in N(v)$ having only black neighbors could
indicate that $u$ itself is black or, alternatively, that $\ENV$ does not follow
$\oBP$ but is nevertheless close to doing so).
To obtain the lower bound we show that this other strategy also requires too
many queries---although it might be more efficient when $\Delta$ is large (thus
explaining why we get a weaker result in that case).

For two-sided error algorithms, we are able to prove similar, though slightly
more modest lower bounds.
These are also based on expander graphs but require a more complex set of
instances for the argument to go through.

\begin{restatable}{theorem}{restateThmLBTwoSided}%
  \label{thm:lb-two-sided}
  There are constants $\eps_0, \zeta > 0$ such that, for any $0 < \eps \le
  \eps_0$ with $\eps \ge \zeta\Delta^b/n$ where $b \ge 1$ is constant, deciding
  if $\ENV \in \oBP$ or is $\eps$-far from $\oBP$ with a \emph{non-adaptive,
  two-sided error} tester requires $q$ queries in general, where:
  \begin{itemize}
    \item If $b > 3$, then $q = \Omega(\Delta/\eps)$.
    \item If $b = 3$, then $q = \Omega(\Delta/\eps\log\Delta)$.
    \item If $1 \le b < 3$, then $q = \Omega(\Delta^{(b-1)/2}/\eps)$.
  \end{itemize}
\end{restatable}

Again focusing on the regime where $\eps$ is constant, we now obtain
$\Omega(\Delta)$ as the lower bound for regimes where $\Delta = O(n^{1/3-c})$
for a constant $c > 0$ or also $\tilde{\Omega}(\Delta)$ when $\Delta =
\Theta(n^{1/3})$.
Hence, given the algorithm of \cref{thm:alg-low-degree}, up to $\Delta =
\Theta(n^{1/3})$ there is essentially \emph{no advantage} for two-sided error
algorithms compared to one-sided error ones.
For larger values of $\Delta$, the lower bound is $\Omega(\sqrt{n/\Delta})$ and
again deteriorates as $\Delta$ increases.

Since we are dealing with two-sided error algorithms, we apply Yao's minimax principle, and we now generate instances
according to two different distributions $D_Y$ and $D_N$ where $D_Y$ follows
$\oBP$ whereas $D_N$ generates instances that are (with high probability) far
from doing so.
The point is that we can show that it is hard to distinguish between $D_Y$ and
$D_N$ without making a considerable number of queries.
The distributions are such that, in both cases, we pick a set $S$ of
$\Theta(\eps n/\Delta)$ vertices in the first time step uniformly at random.
Then we color $S$ and $N(S)$ black in $D_Y$ (and leave the remaining nodes
white) while in $D_N$ we color only \emph{a (constant) fraction} of $N(v)$ for
$v \in S$ black.
(We must also offset the fact that nodes in the second step in $D_N$ are colored
black with less probability by using a larger $S$ when generating $D_N$
instances.)
By the expansion guarantees, this then gives us $\eps$-farness of the instances
in $D_N$.
Observe that in this setting it is meaningless to query nodes in the first step
since only a small fraction of them can ever be black; hence we need only deal
with a set $Q$ of nodes that are queried in the second step.
The indistinguishability of $D_Y$ and $D_N$ follows from using the expansion
guarantee from nodes in the second step to those in the first one.
The argument is that, unless the set $Q$ of queried nodes is large, almost all
neighbors of $Q$ are in fact \emph{unique neighbors} and, moreover, it is
impossible to distinguish $D_Y$ from $D_N$ if the set $S$ only intersects the
unique neighbors of $Q$.
(That is, one can only distinguish $D_Y$ and $D_N$ if one queries \emph{two}
distinct neighbors $u,u' \in N(v)$ of some $v \in S$; due to the expansion
guarantees, this requires a large number of queries.)

In light of these lower bounds, looking back at the algorithm of
\cref{thm:alg-low-degree} we realize that its single weakness is that it does
not perform well when $\Delta$ is large.
Unfortunately our lower bounds do not say as much in that case, and thus a wide
gap is left between lower and upper bounds in that regime.
Nevertheless, we can narrow this gap by using a more complex strategy---if we
are prepared to let go of non-adaptiveness and time-conformability (though we
can still obtain a one-sided error algorithm).
As previously discussed, however, this is not such a large limitation when
taking possible applications into account (as when $T = 2$ it is plausible to,
e.g., require nodes to cache their previous state) and indeed it is offset by
the significant reduction in the query complexity.

\begin{restatable}{theorem}{restateThmAlgLargeDegree}%
  \label{thm:alg-large-degree}
  Let $T = 2$ and let $\eps > 0$ be given.
  There is an \emph{adaptive, one-sided error} algorithm for testing whether
  $\ENV \in \oBP$ or is $\eps$-far from $\oBP$ with query complexity
  \[
    O\left( 
      \frac{\sqrt{n} \log^{3/2} n}{\eps}
    \right).
  \]
\end{restatable}

The algorithm achieving this query complexity is far less trivial than that of
\cref{thm:alg-low-degree}.
Indeed, it must decide whether $\ENV \in \oBP$ or not \emph{without being able
to query the entire neighborhood of any node}.
To achieve this, the algorithm uses a \enquote{filtering} process in which we
first try to infer the color (assuming $\ENV \in \oBP$) of as much nodes as we
can (in either step) by querying some of their neighbors indirectly.
Since we are certain of which color these nodes must have, we can verify these
separately using a small number of random queries.
By some careful observations, we then realize that we can simply ignore these
nodes afterwards and thus reduce the degree of most of the remaining nodes to
$\tilde{O}(\sqrt{n})$.
This allows us to essentially fall back to a strategy as in the algorithm of
\cref{thm:alg-low-degree}, though a particular corner case requires special
attention.

The results for $T = 2$ and the regime where $\eps$ is constant are summarized
in \cref{fig:results-t-2}.

\begin{figure}
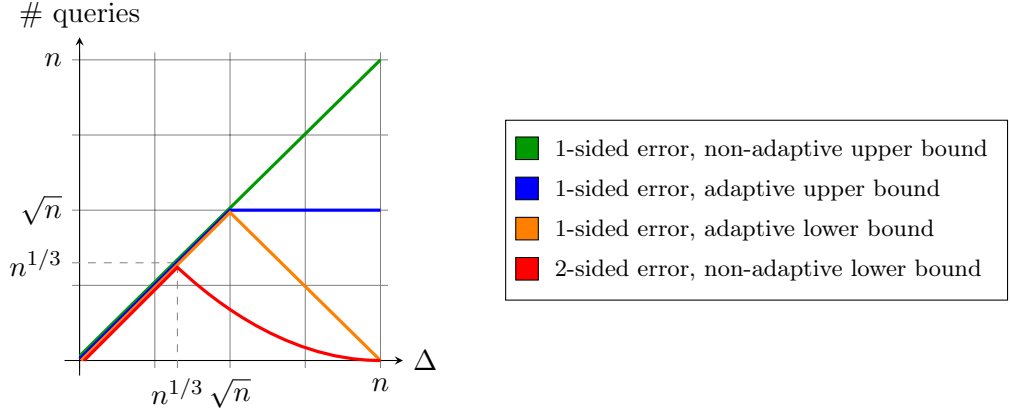

  \centering
  \includestandalone{figs/results_t_2}
  \caption{Summary of results for the case $T = 2$ and constant $\eps$, ignoring
  logarithmic factors}
  \label{fig:results-t-2}
\end{figure}

\subsubsection{Case of General \texorpdfstring{$T$}{T}}

Let us now discuss the case $T > 2$.
Here we obtain a couple of upper bounds that show that the problem admits testing
algorithms with sublinear query complexity, at least in a few regimes of
interest.
We present two algorithms that complement each other and that we discuss next.

We should note at this point that a quick observation shows that the problem
becomes essentially trivial when $T \ge 2\diam(G)/\eps$.
(In a nutshell, this is because otherwise $\ENV \in \oBP$ reaches a fixed point
well before $T$, and thus most configurations of $\ENV$ must all be this one
fixed point.)
Hence for this discussion it should be kept in mind that the problem is only
interesting when $\diam(G)$ is non-trivial and $T = o(\diam(G)/\eps)$.
Furthermore, recall that, since $\ENV$ has $nT$ entries, the benchmark for a
non-trivial testing algorithm is not $o(n)$ but $o(nT)$.

The first algorithm we present is a direct generalization of the one from
\cref{thm:alg-low-degree}.

\begin{restatable}{theorem}{restateThmAlgLowDegreeGeneralT}%
  \label{thm:alg-low-degree-gen-T}
  Let $\eps > 0$ and $T > 2$. 
  There is a \emph{non-adaptive, one-sided error} algorithm that performs
  $O(\Delta^{T-1} / \eps T)$ queries and decides if $\ENV \in \oBP$ or
  $\dist(\ENV,\oBP) \ge \eps$.
\end{restatable}

The algorithm is only useful in settings where, say, $T = O(\log_\Delta n)$.
Nevertheless, it is relatively simple to obtain and outperforms our more complex
algorithm in certain regimes.
\begin{restatable}{theorem}{restateThmAlgLowDiam}%
  \label{thm:alg-low-diam}%
  Let $\eps > 0$ and $T \ge 4/\eps$. 
  Then there is a \emph{non-adaptive, one-sided error} algorithm with query
  complexity $O(\abs{E}\log(n)/\eps T)$ that decides whether $\ENV \in \oBP$ or
  $\dist(\ENV,\oBP) \ge \eps$.
  In addition, if $G$ excludes a fixed minor $H$ (which includes the case where
  $G$ is planar or, more generally, $G$ has bounded genus), then $O(\abs{E}/\eps
  T)$ queries suffice.
\end{restatable}

To better judge what this algorithm achieves, let us suppose that the underlying
graph is $\Delta$-regular, in which case $\abs{E} = n\Delta$.
Then this gives a non-trivial testing algorithm whenever $T =
\omega(\sqrt{(\Delta/\eps) \log n})$ (or $T = \omega(\sqrt{\Delta/\eps})$ if we
also assume $G$ is planar).
Hence, together with \cref{thm:alg-low-degree-gen-T}, we obtain non-trivial
testing algorithms in the regime where $\Delta = o(\log n)$ (or even $\Delta =
o(\log^2 n)$ in planar graphs) and \emph{for all values of $T$}.

The algorithm of \cref{thm:alg-low-diam} combines some ideas from the work of
\textcite{nakar21_back_icalp} with \emph{graph decompositions}.
A graph decomposition is a set $C$ of edges which cuts the graph into components
pairwise disjoint components $V_1,\dots,V_r$ of small diameter.
In our case the appropriate choice of diameter will be $d = O(\eps T)$.
The basic approach is to query the endpoints of $C$ after $d$ steps have
elapsed and then use this view to predict the colors of every node in the graph
in the subsequent steps.
As we show, the view actually suffices to predict all but at most an $O(\eps)$
fraction of $\ENV$ (and hence we need only query the predicted values using
$O(1/\eps)$ independent queries to check if $\ENV$ is following $\oBP$ or not).
We refer to \cref{sec:alg-low-diam} for a more in-depth description of the
strategy and the ideas involved.

\subsection{Open Problems}

Since this work is but a first step in an unexplored direction, several
questions remain open:
\begin{itemize}
  \item \emph{The case $T = 2$ and large $\Delta$.}
  The algorithm of \cref{thm:alg-low-degree} is essentially optimal up to
  $\Delta = O(\sqrt{n})$ (if we consider only one-sided error algorithms), but
  for larger values of $\Delta$ the best we have is the
  $\tilde{O}(\sqrt{n})$-query algorithm of \cref{thm:alg-large-degree}.
  Can we reduce this, say,  to $\tilde{O}(\Delta^{b-1}/\eps)$ for $\eps =
  \Omega(\Delta^b/n)$ so as to match the lower bound of \cref{thm:lb-one-sided}?
  Is it really necessary to give up time-conformity in order to do better than
  $O(\Delta/\eps)$ in this setting?
  In addition, improving our lower bounds in the case of (both one- and
  two-sided error) adaptive algorithms seems well within reach.
  \item \emph{The case $T > 2$.}
  Our results for the case of general $T$ show that we can get non-trivial
  algorithms for graphs of small degree (e.g., $\Delta = o(\log n)$).
  Given the difficulties in the case $T = 2$, it is not surprising that larger
  values of $\Delta$ pose additional challenges.
  In this sense a first step in this direction would be to port the lower bounds
  from the $T = 2$ case.
  Nevertheless, it is not immediately clear how to do so since $\eps$-farness
  there is even harder to achieve given the cascading effects that might occur
  over multiple time steps (see in particular
  \cref{lem:gen-T-relation-viol-dist}).
  \item \emph{Testing other rules.}
  Finally, from a broader perspective it would also be meaningful to consider
  other rules than $\oBP$.
  Of course, by inverting the roles of $0$ and $1$, all of our results also hold
  for AND rule (i.e., a node becomes a $1$ if and only if all its neighbors are
  $1$; otherwise it becomes a $0$).
  Some very natural and interesting rules to consider next are, for instance,
  $\tauBP$ or the majority rule.
  There has been extensive study of these rules in other contexts \cite{%
    gregorio09_bootstrap_ecss,zehmakan19_tight_lata,janson12_bootstrap_aap,%
    maldonado24_local_sofsem,zehmakan21_majority_aaai,balogh07_bootstrap_rsa,%
    gartner18_majority_latin, frischknecht13_convergence_disc,%
    kaaser16_voting_mfcs%
  }, and so there is solid ground to build on there.
\end{itemize}

\subsection{Paper Overview}

The rest of the paper is structured as follows:
In \cref{sec:preliminaries} we introduce basic notation, review some standard
graph-theoretic results that we need, and formally specify the model and problem
we study.
The three sections that follow each cover one part of the results.
On the case $T = 2$ in \cref{sec:alg-T-2} we address the two algorithms 
(\cref{thm:alg-low-degree,thm:alg-large-degree}) and in \cref{sec:lower-bounds}
the two lower bounds (\cref{thm:lb-one-sided,thm:lb-two-sided}).
Finally in \cref{sec:alg-gen-T} we discuss the two algorithms for the case $T >
2$ (\cref{thm:alg-low-degree-gen-T,thm:alg-low-diam}).


\section{Preliminaries}
\label{sec:preliminaries}

The set of non-negative integers is denoted by $\N_0$ and that of strictly
positive integers by $\N_+$.
For $n \in \N_+$, we write $[n] = \{ i \in \N_+ \mid i \le n \}$ for the set of
the first $n$ positive integers.
Without ambiguity, for a statement $S$, we write $[S]$ for the indicator
variable of $S$ (i.e., $[S] = 1$ if $S$ holds; otherwise, $[S] = 0$).

An event is said to occur with high probability if it occurs with probability $1
- o(1)$.
For a set $X$, we write $U_X$ to denote a random variable that takes on values
from $X$ following the uniform distribution on $X$.
We assume the reader is familiar with basic notions of discrete probability
theory (e.g., Markov's inequality and the union bound).
We will use the following version of the Chernoff bound (see, e.g.,
\cite{goldreich08_computational_book,vadhan12_pseudorandomness_book}):

\begin{theorem}[Chernoff bound] \label{thm_chernoff}%
  Let $n \in \N_+$ and $\eps > 0$, and let $X_1,\dots,X_n$ be independent and
  identically distributed random variables taking values in the interval
  $[0,1]$.
  Then, for $X = (\sum_{i=1}^n X_i)/n$ and $\mu = \E[X]$,
  \[
    \Pr\left[ \abs{X - \mu} > \eps \right] < 2 e^{-n \eps^2 / 3 \mu}.
  \]
\end{theorem}

\subsection{Graph Theory}
We consider only undirected graphs.
Except when explicitly written otherwise, we always write just \enquote{graph}
for a simple graph, though self-loops are allowed.

Let $G = (V,E)$ be a graph.
For $S \subseteq V$, $G[S]$ denotes the subgraph of $G$ induced by $S$.
For two nodes $u,v \in V$, $\dist_G(u,v)$ is the length of the shortest path
between $u$ and $v$; we drop the subscript if $G$ is clear from the context.
The \emph{diameter} $\diam(G)$ of $G$ is the maximum length among all shortest
paths between any pair of vertices $u,v \in V$, that is, $\diam(G) = \max_{u,v
\in V} \dist(u,v)$.
This notion extends to any $V' \subseteq V$ by considering only pairs of
vertices in $V'$, that is, $\diam(V') = \max_{u,v \in V'}
\dist(u,v)$.\footnote{%
  This is referred to as the \emph{weak} diameter.
  An alternative notion where we restrict not only the endpoints but also the
  inner vertices of the paths also exists and is called the \emph{strong}
  diameter.
  In this paper we work only with the weak diameter.
}
We write $\delta(G)$ for the minimum degree of $G$ and $\Delta(G)$ for the
maximum one.
If $G$ is clear from the context, we simply write $\delta$ and $\Delta$,
respectively.
If $\delta = \Delta$, then $G$ is \emph{$\Delta$-regular}.


For a node $v \in V$, $N(v) = \{ u \in V \mid uv \in E \}$ denotes the set of
\emph{neighbors} of $v$.
Generalizing this notation, for a set $S \subseteq V$ we write $N(S)$ for the
union $\bigcup_{v \in S} N(v)$.
A vertex $u \in V$ is said to be a \emph{unique neighbor} of $S$ if there is a
\emph{unique} $s \in S$ such that $us \in E$.
When $S$ is clear from the context, we also refer to a unique neighbor of $v \in
S$ as a node $u \in V$ for which $u \in N(v')$ if and only if $v' \notin S$ or
$v' = v$.

A graph $G = (V,E)$ is \emph{bipartite} if $V = L \cup R$ for disjoint sets $L$
and $R$ and any edge has exactly one endpoint in $L$ and one in $R$.
In this context, we refer to the nodes of $L$ as \emph{left-} and to those of
$R$ as \emph{right-vertices}.
Additionally, the graph is \emph{balanced} if $\abs{L} = \abs{R}$.

The following is a spin-off of a well-known result on the size of the dominating
set of a graph (see, e.g., \cite{alon08_probabilistic_book}):

\begin{lemma}[Cover from minimum degree]%
  \label{lem:dom-set}
  Let $G = (V,E)$ be a bipartite graph where each right-vertex has degree at
  least $\delta$.
  Then there is a set $D$ of $n \log(n)/\delta$ left-vertices such that every
  right-vertex has a neighbor in $D$.
\end{lemma}

\begin{proof}
  We prove the claim using the probabilistic method.
  Fix a right-vertex $v \in V$.
  If we pick a set $D$ of $m = n \log(n)/\delta$ left-vertices uniformly at
  random, then the probability that $N(v) \cap D$ is empty is at most
  $( 1 - \delta/n )^m < e^{-\log n} < 1/n$.
  Hence, by the union bound, there is a non-zero probability that $D$ is such
  that $N(v) \cap D$ is non-empty for every right-vertex $v$.
\end{proof}

\paragraph{Expander graphs.}
In general, a graph $G = (V,E)$ is said to be an \emph{expander graph} if we
have $\abs{N(S)} \ge (\Delta - r) \abs{S}$ for every set $S \subseteq V$ where
$\abs{S} \le K$, for particular values of $r$ and $K$.
(Ideally $K$ is as large and $r$ as small as possible.)
Expander graphs have found ample application in diverse areas of theoretical
computer science (see, e.g., \cite{hoory06_expander_bams} for a survey), and
property testing is no different \cite{bhattacharyya22_property_book}.
We will need a couple of consequences of this property.

\begin{lemma}[Unique neighbors from expansion]
  \label{lem_expanders_unique_neighbors}%
  Let $G = (V,E)$ be a graph and $S \subseteq V$ be a set with $m$ outgoing
  edges and $\abs{N(S)} \ge (1-\alpha) m$ for some $\alpha \ge 0$.
  Then $S$ has at least $(1-2\alpha) m$ unique neighbors.
\end{lemma}

\begin{proof}
  Let $U$ be the set of unique neighbors of $S$ and $B = N(S) \setminus U$.
  Since the nodes in $B$ have at least $2$ incident edges originating from $S$,
  we have
  \[
    m \ge \abs{U} + 2 \abs{B} = 2 \abs{N(S)} - \abs{U} \ge 2(1-\alpha)m -
    \abs{U}.
  \]
  Solving for $\abs{U}$ yields the statement.
\end{proof}

\begin{restatable}[Upper bounds on number of common neighbors]%
  {lemma}{restateLemExpandersIntersection}%
  \label{lem_expanders_intersection}%
  Let $G = (V,E)$ be a graph and $\alpha \le 1/2$ be such that $\abs{N(S)} \ge
  (1-\beta)\Delta\abs{S}$ holds for every $\abs{S} \le \alpha n$.
  In addition, let \emph{disjoint} subsets $S,S' \subseteq V$ be given.
  The following holds:
  \begin{enumerate}
    \item If $\abs{S} + \abs{S'} \le \alpha n$, then
    \[
      \abs{N(S) \cap N(S')} \le \beta \Delta (\abs{S} + \abs{S'}). 
    \]
    \item Let $\beta < 1/2$, $\abs{S} \le \alpha n / 2$, and $\abs{S'} \ge
    \abs{S}$.
    Then
    \[
      \Pr_{v \in U_{S'}}\left[ \abs*{N(v) \cap N(S)} > 2\beta\Delta \right]
      \le \frac{\abs{S}}{\abs{S'}};
    \]
    that is, for all but at most $\abs{S}$ nodes $v \in S'$, $\abs{N(v) \cap
    N(S)} \le 2\beta\Delta$.
  \end{enumerate}
\end{restatable}

\begin{proof}
  \begin{enumerate}
    \item By the expansion property, we have
    \[
      \abs{N(S) \cup N(S')} = \abs{N(S \cup S')}
      \ge (1-\beta) \Delta (\abs{S} + \abs{S'}).
    \]
    On the other hand, we can upper-bound the left-hand side as follows:
    \[
      \abs{N(S) \cup N(S')} \le \abs{N(S)} + \abs{N(S')} - \abs{N(S) \cap N(S')}
      \le \Delta (\abs{S} + \abs{S'}) - \abs{N(S) \cap N(S')}.
    \]
    Combining the two inequalities and solving for $\abs{N(S) \cap N(S')}$ yields
    the statement.
    \item Let $S'' \subseteq S'$ be arbitrary with $\abs{S''} = \abs{S}$.
    Again, by the expansion property,
    \[
      \abs{N(S) \cup N(S'')} = \abs{N(S \cup S'')}
      \ge 2(1-\beta)\Delta\abs{S}.
    \]
    Meanwhile we have $\abs{N(S)} \le \Delta\abs{S}$, which means there are at
    least $(1-2\beta)\Delta\abs{S}$ nodes in $N(S'') \setminus N(S)$.
    By averaging, there is at least one node $v \in S''$ such that $\abs{N(v)
    \cap N(S)} \le 2\beta\Delta$.
    Since the argument applies to an arbitrary subset $S''$ of $S'$, by
    excluding each such $v$ one by one it follows that all but at most $\abs{S}$
    nodes in $S'$ have this property.
    \qedhere
  \end{enumerate}
\end{proof}

\subsection{Model and Problem Definition}

We use the standard query model of property testing
\cite{bhattacharyya22_property_book}.
The testing algorithm has unlimited computational power and access to a source
of infinitely many random bits that are fully independent from one another.
In addition, we give the model full knowledge of the underlying topology of the
network, which is presented as a graph $G = (V,E)$ with $\abs{V} = n$ nodes.
We assume there are no singleton nodes (i.e., every node is such that there is
an edge incident to it).
The topology remains fixed during the evolution of the network, whose nodes take
on different states over a set of discrete time steps.
As in the previous works \cite{nakar21_back_icalp,goldreich17_learning_jacm},
the formal object we are testing is an \emph{environment} $\ENV\colon V \times
[T] \to Z$ where $T \ge 2$ and $Z$ is the set of states that each node may
assume.

The goal is to detect whether $\ENV$ is following a certain \emph{local rule}
$\rho$, which is defined as a function that maps every multiset $\mu$ over $Z$
to $\rho(\mu) \in Z$.
The environment $\ENV$ is said to \emph{follow} $\rho$ if, for every time step
$t \le T$ and every node $v \in V$, we have that $\ENV(v,t+1) =
\rho(\ENV(N(v),t))$ (where $\ENV(N(v),t)$ here is seen as a multiset, that is,
counting multiplicities of the occurrence of each element of $Z$).
Blurring the distinction between $\rho$ and the set of environments that follow
it, we write $\ENV \in \rho$ if $\ENV$ follows $\rho$.

The \emph{distance} between two environments $\ENV, \ENV'\colon V \times [T] \to
Z$ is the (normalized) number of pairs on which $\ENV$ and $\ENV'$ differ:
\[
  \dist(\ENV,\ENV') = \frac{1}{nT}
    \sum_{(v,t) \in V \times [T]} [\ENV(v,t) \neq \ENV'(v,t)].
\]
For a set of environments $X$ (all over the same domain $V \times [t]$), we
write
\[
  \dist(\ENV,X) = \min_{\ENV' \in X} \dist(\ENV,\ENV')
\]
for the minimum distance between $\ENV$ and $X$.
Being a bit sloppy, we write $\dist(\ENV,\rho)$ for the minimum distance between
$\ENV$ and the set of environments $\ENV'$ for which $\ENV' \in \rho$.
For $\eps \ge 0$, $\ENV$ is said to be \emph{$\eps$-far} from $\rho$ if
$\dist(\ENV,\rho) \ge \eps$; otherwise $\ENV$ is \emph{$\eps$-near} $\rho$.

In this work, we focus on $Z = \binalph$ and on testing the $\oBP$ rule of
bootstrap percolation.
The rule is defined by $\rho(\mu) = [1 \in \mu]$ (i.e., $\rho(\mu) = 1$ if $1
\in \mu$ and $\rho(\mu) = 0$ otherwise).
Seeing states as colors, we identify state $1$ with the color \emph{black} and
state $0$ with \emph{white}.
(Being pedantic, the $\oBP$ rule in the context of bootstrap percolation is such
that a black node always remains black.
This behavior can be enforced in the model we describe by adding self-loops to
all nodes.)

For $t \ge 2$, a pair $(v,t)$ is a \emph{successor} of $(u,t-1)$ if there is an
edge between $v$ and $u$; at the same time, $(u,t-1)$ is a \emph{predecessor} of
$(v,t)$.
If the respective time steps $t$ and $t-1$ are clear from the context, we might
also drop any mention of them and simply say that $v$ (as a node) is a successor
of $u$.
This is particularly convenient when analyzing the case $T=2$.

\paragraph{Testing algorithms.}
Fix $\eps > 0$.
A \emph{testing algorithm} $A$ for $\oBP$ accesses $\ENV\colon V \times [T] \to
Z$ by means of \emph{queries}, which are pairs $(v,t) \in V \times [T]$. 
Upon querying the pair $(v,t)$, $A$ receives $\ENV(v,t)$ as answer.
If the queries are performed in an order where, for every $t$ and $t'$ with $t'
> t$, $A$ never makes a $(\cdot,t)$ query after it has queried $(\cdot,t')$,
then $A$ is said to be \emph{time-conforming}.
As usual in property testing, our interest lies in the \emph{query complexity}
of $A$, that is, the maximum number of queries that $A$ makes, regardless of its
randomness.
The algorithm $A$ is a \emph{one-sided error tester} for $\ENV \in \oBP$ if the
following holds, where the probabilities are taken over the randomness of $A$:
\begin{itemize}
  \item If $\ENV \in \oBP$, then always $A(\ENV) = 1$.
  \item If $\ENV$ is $\eps$-far from $\oBP$, then $\Pr[A(\ENV) = 1] < 1/2$.
\end{itemize}
In contrast, $A$ is a \emph{two-sided error tester} if it may also err on $\ENV
\in \oBP$:
\begin{itemize}
  \item If $\ENV \in \oBP$, then $\Pr[A(\ENV) = 1] \ge 2/3$.
  \item If $\ENV$ is $\eps$-far from $\oBP$, then $\Pr[A(\ENV) = 1] < 1/3$.
\end{itemize}

\paragraph{Violations.}
Observe that our notion of distance is \emph{not} the same as counting the
number of failures of $\ENV$ in following $\oBP$.
There are two kinds of failures that may occur:

\begin{definition}[Violations]
  A pair $(v,t) \in V \times [T]$ is \emph{violating} if $t \ge 2$ and one of
  the following conditions hold:
  \begin{enumerate}[(I)]
    \item $\ENV(v,t) = 0$ and $\exists u \in N(v): \ENV(u,t-1) = 1$
    \item $\ENV(v,t) = 1$ and $\forall u \in N(v): \ENV(u,t-1) = 0$
  \end{enumerate}
  We refer to these violations as \emph{violations of type I and II},
  respectively.
  We write $\viol(\ENV)$ for the set of violating pairs in $\ENV$.
\end{definition}

\begin{figure}
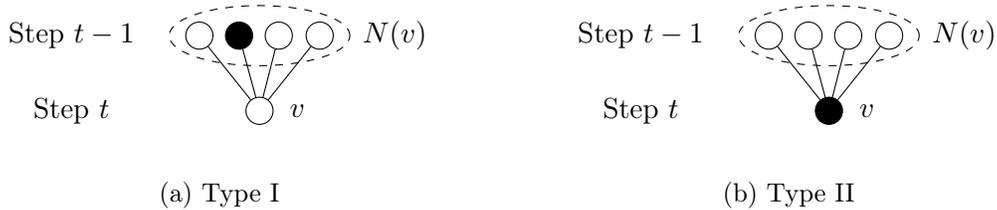

  \centering
  \hspace*{\fill}
  \subcaptionbox{Type I}{\includestandalone{figs/violation_type_I}}
  \hspace*{\fill}
  \subcaptionbox{Type II}{\includestandalone{figs/violation_type_II}}
  \hspace*{\fill}
  \caption{Violations can be of two different types.
    Here we see a node $v$ and its state in time step $t$ (as a color) as well
    as its neighbors $N(v)$ and their respective states in step $t-1$.}
\end{figure}

Although a larger distance to $\oBP$ implies a greater number of violations,
there is not an exact correspondence between the two.
For example, it might be the case that $\ENV$ exhibits a great number of
violations, but correcting them requires recoloring only a few nodes.
We will prove upper and lower bounds between the distance and the number of
violations further below
(\cref{lem:meq2_viol_vs_dist,lem:gen-T-relation-viol-dist}).


\section{Upper Bounds for the Case \texorpdfstring{$T=2$}{T = 2}}
\label{sec:alg-T-2}

In this section we present our two algorithms for the case where $T = 2$.
The first of these (\cref{sec:upper-bound-delta-eps}) is quite simple and has
query complexity $O(\Delta/\eps)$, which turns out to be optimal for the regimes
where $\Delta = o(\sqrt{n})$.
The second one (\cref{sec:upper-bound-sqrt-n-eps}) is much more intricate and
gives query complexity $\tilde{O}(\sqrt{n}/\eps)$, which makes it more suitable
for the regimes where $\Delta = \omega(\sqrt{n})$.
Although both are one-sided error algorithms, the first algorithm is
non-adaptive and thus time-conforming whereas the second has neither of these
properties (i.e., it is adaptive and also does not respect time-conformity).

\subsection{An Upper Bound that Scales with the Maximum Degree}
\label{sec:upper-bound-delta-eps}

In this section, we prove:

\restateThmAlgLowDegree*

The claim is that \cref{alg:m2} satisfies the requirements of
\cref{thm:alg-low-degree}.
As mentioned above, the strategy followed by \cref{alg:m2} is quite simple:
It chooses a certain subset of nodes $U$ uniformly at random and then queries
the states of $u \in U$ and all of $N(u)$ in both time steps.
The algorithm then rejects if and only if a violation of either type is
detected.

\begin{algorithm}
  Pick $U \subseteq V$ uniformly at random where $\abs{U} = \ceil{2/\eps}$\;
  Query $\ENV(v,1)$ and $\ENV(u,2)$ for every $u \in U$ and $v \in N(u)$ in a
  time-conforming manner\;
  \For{$u \in U$}{
    \lIf{$\ENV(u,2) = 0$ and $\exists v \in N(u): \ENV(v,1) = 1$}{
      \Reject
    }
    \lIf{$\ENV(u,2) = 1$ and $\forall v \in N(u): \ENV(v,1) = 0$}{
      \Reject
    }
  }
  \Accept\;
  \caption{Algorithm for the case $T = 2$ with query complexity
  $O(\Delta/\eps)$}
  \label{alg:m2}
\end{algorithm}

At the core of the correctness of \cref{alg:m2} is the relation between the
number of violations and the distance of $\ENV$ to $\oBP$.
With a bit of care, we can relate the two quantities as shown next.
(Actually for the correctness of \cref{alg:m2} we only need one of the two
bounds below; the other one comes as a \enquote{bonus}.)

\begin{lemma}
  \label{lem:meq2_viol_vs_dist}
  Let $T = 2$.
  Then
  \[
    \frac{\abs{\viol(\ENV)}}{2\Delta n}
    \le \dist(\ENV,\oBP)
    \le \frac{\abs{\viol(\ENV)}}{2n}.
  \]
\end{lemma}


\begin{proof}
  Every violating pair $(u,t)$ can be corrected by flipping the value of
  $\ENV(u,t)$, which does not create a new violating pair since $t = T = 2$.
  In addition, if $\ENV$ does not have any violating pair, then $\ENV \in \oBP$.
  This implies $\dist(\ENV,\oBP) \le \abs{\viol(\ENV)} / 2n$.
  On the other hand, flipping the color of a node can only correct at most
  $\Delta$ violating pairs.
  Hence we also have $\dist(\ENV,\oBP) \ge \abs{\viol(\ENV)} / 2 \Delta n$.
\end{proof}

The lemma directly implies that, if $\dist(\ENV,\oBP) \ge \eps$, then
$\abs{\viol(\ENV)} \ge 2 \eps n$.
Hence the probability that \cref{alg:m2} errs in this case is 
\[
  \Pr[(U,2) \cap \viol(\ENV) = \varnothing]
  \le (1 - 2\eps)^{\abs{U}}
  < \frac{1}{e}
  < \frac{1}{2}.
\]
Since \cref{alg:m2} only rejects when a violation of either type is detected,
\cref{alg:m2} always accepts if $\ENV \in \oBP$.
The query complexity and other properties of \cref{alg:m2} are clear, and hence
\cref{thm:alg-low-degree} follows.

\subsection{An Upper Bound Independent of the Maximum Degree}
\label{sec:upper-bound-sqrt-n-eps}

Next we show our second algorithm, which is much more complex than
\cref{alg:m2}.
Since \cref{alg:m2} is already optimal for $\Delta = O(\sqrt{n})$, we focus on
the regime where $\Delta = \Omega(\sqrt{n})$ and present an algorithm with query
complexity that is independent of $\Delta$.
The algorithm requires adaptiveness and unfortunately is no longer
time-conforming; obtaining a time-conforming or even non-adaptive algorithm with
the same query complexity for these large values of $\Delta$ (or proving none
exists) remains an interesting open question.

\restateThmAlgLargeDegree*

We claim \cref{alg:large-degree} satisfies the requirements of the theorem.
Next we give a brief description of the strategy followed by
\cref{alg:large-degree}.

\begin{algorithm}
  Select $Q_1,Q_1',Q_2,Q_2' \subseteq V$ with $\abs{Q_i} = \abs{Q_i'} =
  (24/\eps)\sqrt{n}\log^{3/2} n$ uniformly at random\;
  Query $\ENV(Q_1,1)$ and $\ENV(Q_2',1)$\;
  $B_2 \gets \{ v \in V \mid \exists u \in N(v)\cap Q_1\colon \ENV(u,1) = 1\}$\;
  Query $\ENV(Q_1',2)$ and $\ENV(Q_2,2)$\;
  $W_1 \gets \{ v \in V \mid \exists u \in N(v)\cap Q_2\colon \ENV(u,2) = 0\}$\;
  \lIf{$\exists u \in Q_1' \cap B_2: \ENV(u,2) = 0$
      or $\exists u \in Q_2' \cap W_1: \ENV(u,1) = 1$}
    {\Reject}
  \label{line:alg-large-degree-first-reject}
  $F \gets \{ v \in V \mid \abs{N(v) \setminus W_1} \le 4\sqrt{n\log n} \}$\;
  Select $Q_3 \subseteq F$ with $\abs{Q_3} = (4/\eps)\log n$ uniformly at
  random\;
  Query $\ENV(v,2)$ and $\ENV(N(v) \setminus W_1,1)$ for every $v \in Q_3$\;
  \lIf{$\exists v \in Q_3: \ENV(v,2) = 0
      \land \exists u \in N(v) \setminus W_1: \ENV(u,1) = 1$
      or $\exists v \in Q_3: \ENV(v,2) = 1
      \land \nexists u \in N(v) \setminus W_1: \ENV(u,1) = 1$}
    {\Reject}
  \label{line:alg-large-degree-second-reject}
  \Accept\;
  \caption{Algorithm for the case $T = 2$ with query complexity
  $\tilde{O}(\sqrt{n}/\eps)$}
  \label{alg:large-degree}
\end{algorithm}

\paragraph{Approach.}

The operation of \cref{alg:large-degree} can be divided into two parts.
The first one is up to \cref{line:alg-large-degree-first-reject}.
Here we query nodes from the first and second time steps at random ($Q_1$ and
$Q_2$) and try to ascertain the color of as many nodes as possible using these
queries.
More specifically, if a node $v$ has a neighbor $u \in N(v)$ which is black in
the first step, then we know $v$ must be black in the second step.
We gather these nodes in the set $B_2$.
A similar observation holds for the nodes in the set $W_1$, which must be white
since they have a neighbor in the second step that is white.
At the same time we query another set of nodes from the first and second step
uniformly at random ($Q_1'$ and $Q_2'$) to verify that all but a very small
fraction of nodes in $W_1$ (resp., $B_2$) are indeed white (resp., black).

The second part of the algorithm starts after
\cref{line:alg-large-degree-first-reject}.
Here we will ignore nodes in $W_1$ (since we already know they are white) and
\enquote{filter} nodes that have not too large degree to nodes not in $W_1$.
These nodes are added to the set $F$.
Intuitively we can then test these nodes in the same fashion as \cref{alg:m2}:
We select a few nodes $v \in F$ uniformly at random ($Q_3$) and then query the
entire neighborhood of these nodes in the first step, so $\ENV(u,1)$ for $u \in
N(v) \setminus W_1$, as well as $\ENV(v,2)$.
If any violations are detected here, then we can safely reject.
What then remains are only nodes with high degree; as we argue in the analysis
below, any set of nodes not in $F$ that are black (which might occur when $\ENV
\notin \oBP$) and which have no white predecessor can actually be covered by
recoloring only a small set of nodes (and hence $\ENV$ must be close to $\oBP$).

\paragraph{Analysis.}

The query complexity of \cref{alg:large-degree} is clear, so we focus on the
analysis on its correctness.
First we show that \cref{alg:large-degree} is indeed a one-sided error
algorithm; that is:

\begin{claim}
  If $\ENV \in \oBP$, then \cref{alg:large-degree} always accepts.
\end{claim}

Intuitively this is the case because we are only trying to detect violations
(and accept unconditionally if we do not manage to find any).

\begin{proof}
  Since $\ENV \in \oBP$, we have $\ENV(W_1,1) = 0$ and $\ENV(B_2,2) = 1$.
  As a result, \cref{alg:large-degree} never rejects in
  \cref{line:alg-large-degree-first-reject}.
  Consider the two possibilities for \cref{alg:large-degree} to reject in
  \cref{line:alg-large-degree-second-reject}.
  The first is that there is a node $v \in Q_3$ with $\ENV(v,2) = 0$ and some $u
  \in N(v)$ so that $\ENV(u,1) = 1$, which contradicts $\ENV \in \oBP$.
  The second is that $\ENV(v,2) = 1$ and $\ENV(u,1) = 0$ for every $u \in N(v)
  \setminus W_1$; however, since $\ENV(W_1,1) = 0$, this means $\ENV(u,1) = 0$
  for every $u \in N(v) \cap W_1$ as well and then $\ENV(u,1) = 0$ for every $u
  \in N(v)$, thus also contradicting $\ENV \in \oBP$.
\end{proof}

Now we turn to proving that \cref{alg:large-degree} does not have false
positives.
More specifically we show that \cref{alg:large-degree} can only accept with
constant probability if $\dist(\ENV,\oBP) < \eps$ is the case (and so,
conversely, \cref{alg:large-degree} rejects with high probability if 
$\dist(\ENV,\oBP) \ge \eps$).

For $v \in V$ and $t \in \{1,2\}$, we write $\bn_t(v)$ and $\wn_t(v)$ for the
number of black and white neighbors, respectively, of $v$ in step $t$; formally,
\begin{align*}
   \bn_t(v) &= \abs*{\left\{ u \in N(v) \mid \ENV(u,t) = 1 \right\}},
  &\wn_t(v) &= \abs*{\left\{ u \in N(v) \mid \ENV(u,t) = 0 \right\}}.
\end{align*}
Let $\theta = (\eps/4)\sqrt{n/\log n}$ and consider the following sets:
\begin{align*}
   X_1 &= \left\{ v \in V \mid \wn_2(v) \ge \theta \right\},
  &X_2 &= \left\{ v \in V \mid \bn_1(v) < \theta \land \ENV(v,2) = 0 \right\},
  \\
   Y_1 &= \left\{ v \in V \mid \wn_2(v) < \theta \land \ENV(v,1) = 1 \right\},
  &Y_2 &= \left\{ v \in V \mid \bn_1(v) \ge \theta \right\}, 
  \\
   Z_1 &= \left\{ v \in V \mid \wn_2(v) < \theta \land \ENV(v,1) = 0 \right\},
  &Z_2 &= \left\{ v \in V \mid \bn_1(v) < \theta \land \ENV(v,2) = 1 \right\}.
\end{align*}
The sets $X_1$ and $Y_2$ contain the nodes for which we can detect that they
must be white and black, respectively, by using the query sets $Q_1$ and $Q_2$.

\begin{figure}
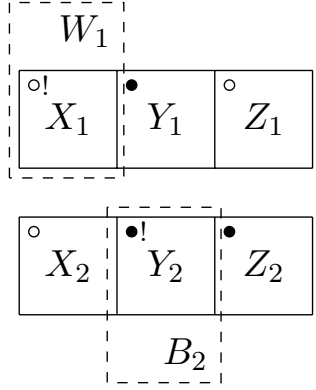

  \centering
  \includestandalone[scale=1.3]{figs/sets_alg_large_degree}
  \caption{Relation between the sets used in the analysis of
  \cref{alg:large-degree}.
  The sets form a partition of the nodes in the first and second time steps.
  The small circles indicate the color of the nodes in each set or, in the case
  of $X_1$ of $Y_2$, that the algorithm rejects unless (almost all) nodes in the
  set have the respective color (denoted with an exclamation mark).}
\end{figure}

\begin{claim}%
  \label{claim:x1-in-w1-and-y2-in-b2}
  With high probability over the choices made by \cref{alg:large-degree}, $X_1
  \subseteq W_1$ and $Y_2 \subseteq B_2$.
\end{claim}

\begin{proof}
  Fix a node $v \in X_1$.
  By the Chernoff bound, the probability that no $u \in N(v)$ with $\ENV(u,2) =
  0$ lands in $Q_2$ is at most $2e^{-2\log n} < 2/n^2$.
  Hence, by the union bound, the probability that there is a node in $X_1$ and
  not in $W_1$ is $O(1/n)$.
  The same argument applies to $Y_2$ and $B_2$.
\end{proof}

Next we observe that the queries from $Q_1'$ and $Q_2'$ significantly
\enquote{reduce} the number of black or white nodes in $X_1$ or $Y_2$,
respectively; that is, if there is a significant number of such nodes in these
sets, then \cref{alg:large-degree} will detect them anyway and reject (and thus
we can focus the analysis on instances where this is not the case).

\begin{claim}%
  \label{claim:x1-y2}
  If $\ENV$ is such that there are $\eps\sqrt{n}$ nodes $v \in X_1$ with
  $\ENV(v,1) = 1$ or $\eps\sqrt{n}$ nodes $v \in Y_2$ with $\ENV(v,2) = 0$, then
  \cref{alg:large-degree} rejects $\ENV$ with high probability.
\end{claim}

\begin{proof}
  Let $S \subseteq V$ be a subset of $\abs{S} \ge \eps\sqrt{n}$ vertices.
  Then the probability that $S \cap Q_i'$ is empty is at most
  \[
    \left( 1 - \frac{\eps}{\sqrt{n}} \right)^{(24/\eps)\sqrt{n}\log^{3/2}n}
    < e^{-24\log^{3/2}n}
    = o\left( \frac{1}{n} \right).
  \]
  Using \cref{claim:x1-in-w1-and-y2-in-b2}, we have $X_1 \subseteq W_1$ and $Y_2
  \subseteq B_2$ with high probability.
  In this case \cref{alg:large-degree} rejects if any node $v \in X_1 \subseteq
  W_1$ with $\ENV(v,1) = 1$ lands in $Q_2'$ or any $v \in Y_2 \subseteq B_2$
  with $\ENV(v,2) = 0$ lands in $Q_1'$.
  Therefore \cref{alg:large-degree} rejects with high probability if there are
  at least $\eps\sqrt{n}$ nodes of either type.
\end{proof}

Hence we may now safely assume that all but at most $O(\eps\sqrt{n})$ nodes in
$X_1$ are white in the first time step and that all but at most
$O(\eps\sqrt{n})$ nodes in $Y_2$ are black in the second one.
The next observation is that nodes in $X_2$ are highly connected to $X_1$.
This justifies filtering nodes based on their connections to $W_1 \supseteq
X_1$.

\begin{claim}%
  \label{claim:avg-y1-x2}
  On average, a node from $X_2$ has at most $\theta n/\abs{X_2}$ neighbors not
  in $X_1$.
\end{claim}

\begin{proof}
  Every node in $Y_1$ or $Z_1$ has at most $\theta$ white neighbors by
  definition, so at most this many neighbors in $X_2$.
  Hence there are at most $\theta n$ edges in total between $X_2$ and nodes not
  in $X_1$.
\end{proof}

Finally we show that, if \cref{alg:large-degree} accepts $\ENV$ with at least
constant probability, then we can correct all violations of either type with at
most $\eps n/2$ modifications in total for each type.
In both cases we must be careful so that these modifications do not create new
violations of their own.

\begin{claim}%
  \label{claim:type-I-viol}
  If \cref{alg:large-degree} accepts $\ENV$ with at least constant probability,
  then there are at most $\eps n/2$ many type I violations in $\ENV$.
  These violations can be corrected (without creating any new ones) by
  recoloring $\ENV(v,2)$ black for every violation $(v,2)$.
\end{claim}

\begin{proof}
  Let $R$ be the set of nodes corresponding to type I violations, that is,
  \[
    R = \{ v \in V \mid \ENV(v,2)=0 \land \exists u \in N(v): \ENV(u,1)=1 \}.
  \]
  We prove the claim by proving the contrapositive; that is, if $\abs{R} \ge
  \eps n/2$, then \cref{alg:large-degree} rejects $\ENV$ with high probability.

  The first observation is that we have $\abs{R \setminus X_2} = o(\eps n)$
  (with high probability) due to \cref{claim:x1-y2} and then, by the assumption
  on $R$, $\abs{X_2} \ge (1-o(1))\eps n/2$.
  Hence we focus our analysis on $R \cap X_2$.
  By \cref{claim:avg-y1-x2}, on average a node from $X_2$ has at most
  $2\theta/\eps = (1/2)\sqrt{n/\log n}$ many neighbors that are outside $X_1$.
  By Markov's inequality, this gives us that there are at most $O(n/\log n)$
  many nodes $v \in X_2$ for which $\abs{N(v) \setminus X_1} > 4\sqrt{n\log n}$.
  Using \cref{claim:x1-in-w1-and-y2-in-b2}, we have $X_1 \subseteq W_1$ and so
  altogether we have $\abs{R \cap F} \ge \eps n/4$ (with high probability).
  In this case the probability that $R \cap Q_3$ is empty is at most
  $(1-\eps/4)^{4\log(n)/\eps} < e^{-\log(n)} = O(1/n)$, and so
  \cref{alg:large-degree} rejects with high probability.
\end{proof}

\begin{claim}
  If \cref{alg:large-degree} accepts $\ENV$ with at least constant probability,
  then all type II violations in $\ENV$ can be corrected by recoloring at most
  $\eps n/2$ nodes.
  In particular, this recoloring is such that we color $\ENV(v,1)$ and
  $\ENV(N(v),2)$ black for a certain subset of nodes $v$ (and hence does create
  any new violations).
\end{claim}

\begin{proof}
  Similar to the proof of \cref{claim:type-I-viol}, let
  \[
    R = \{ v \in V \mid \ENV(v,2)=1 \land \forall u \in N(v): \ENV(u,1)=0 \}.
  \]
  be the set of type II violations.
  We show that, if \cref{alg:large-degree} accepts with at least constant
  probability, then we can correct $R$ by recoloring at most $\eps n/2$ nodes
  black.
  (Note this does not necessarily mean that $\abs{R} < \eps n/2$ as in the proof
  of \cref{claim:type-I-viol}.
  Instead what we prove is an upper bound on the number of recolorings needed to
  correct $R$.)

  By \cref{claim:x1-y2}, $\abs{R \setminus Z_2} = o(\eps n)$ and thus $\abs{Z_2}
  \ge (1-o(1))\eps n/2$.
  Arguing as in the proof of \cref{claim:type-I-viol}, if $\abs{R \cap F} \ge
  \eps n/4$, then \cref{alg:large-degree} must reject with high probability.
  Hence let us focus on the nodes in $R' = R \setminus F$.
  Consider the bipartite graph where the set of left-vertices is $V \setminus
  W_1$, that of right-vertices is $R'$, and the edges are as in $G$.
  Then the minimum degree of $R'$ in this graph is $4\sqrt{n \log n}$, which
  means we can apply \cref{lem:dom-set} and obtain a cover $D \subseteq V
  \setminus W_1$ of $R'$ with $\abs{D} = (1/4)\sqrt{n\log n}$ nodes.
  By \cref{claim:x1-in-w1-and-y2-in-b2}, $D \cap X_1 = \varnothing$ and hence
  $\wn_2(v) < \theta$ for every $v \in D$ (with high probability).
  Therefore we can correct $R'$ by coloring $\ENV(D,1)$ and $\ENV(N(D),2)$ all
  black, which means coloring at most $(\theta/4)\sqrt{n\log n} \le \eps n/16$
  nodes black.
  Together with $\abs{R \cap F} < \eps n/4$, this means we must color at most
  $(1/4+1/16)\eps n < \eps n/2$ many nodes black in total in order to correct
  $R$.
\end{proof}

This concludes the proof of \cref{thm:alg-large-degree}.


\section{Lower Bounds for the Case \texorpdfstring{$T = 2$}{T = 2}}
\label{sec:lower-bounds}

In this section, we prove our two lower bounds for the case $T = 2$, one for
one-sided error (\cref{sec:lb-one-sided}) and the other for two-sided error
algorithms (\cref{sec:lb-two-sided}).

Both proofs are based on an appropriate construction of expander graphs.
It is well-known that random graphs are good expanders; one of the main
challenges here is choosing adequate parameters so we get the properties needed
for obtaining the lower bounds.
Since the construction is the same for both lower bounds, we will address it
first.

We will work with (balanced) bipartite expanders where we have different
expansion guarantees on the left and on the right.
This is needed because, in the hard instances we produce, we need one kind of
expansion to guarantee $\eps$-farness while expansion in the other direction
(possibly in conjunction with the former) drives up the query complexity of a
correct algorithm.
In particular, for the first of these we need rather larger sets to expand but
are satisfied with a smaller expansion rate (and vice-versa for the second one).

The core of our construction is given by the following lemma, which allows us to
obtain expansion in either direction from a \enquote{prototypical} construction.
The result is such that we can \enquote{fine-tune} the upper bound on the size
of the sets that expand (given by $\gamma$ below) and the expansion ratio (which
depends on $\rho$).
Since the construction is symmetrical and satisfies the required property with
$> 1/2$ probability, in particular there is a non-zero probability that it holds
simultaneously in both directions, even for different choices of parameters (as
long as we only ask for one expansion guarantee on either side).

\begin{lemma}\label{lem_expander_lower_bound_one_step_construction}%
  Let $n$ and $\Delta$ be fixed parameters with $\Delta = o(n)$, and also let
  $\rho > 1$ (that may depend on $n$ and $\Delta$) and
  \[
    \gamma \le \min\left\{
      \frac{1}{2},
      \left(
        \frac{1}{2e}
        \left( \frac{\rho}{2 e \Delta^2} \right)^\rho
      \right)^{1/(\rho-1)}
    \right\}
  \]
  be given with $\gamma \ge 2/n$.
  Suppose we sample a random balanced bipartite multigraph $\mathcal{M} = (V,E)$
  where $V = L \cup R$ has $n$ vertices on each side and $E$ is the union of
  $\Delta$ many perfect matchings $M_1,\dots,M_\Delta$ that are chosen
  independently and uniformly at random.
  Then with $> 1/2$ probability the following holds for every set $S \subseteq
  [n]$ of left-vertices (or, by symmetry, right-vertices) of $\mathcal{M}$ with
  $\abs{S} \le \gamma n$:
  \[
    \abs{N(S)} \ge (\Delta - \rho) \abs{S}.
  \]
  Furthermore, $\mathcal{M}$ can be made into a $\Delta$-regular (simple) graph
  $G$ by redistributing edges.
\end{lemma}

The proof is by a probabilistic argument that is reminiscent of the standard
argument showing the existence of bipartite expanders (as in, e.g.,
\cite{hoory06_expander_bams} or also \cite{vadhan12_pseudorandomness_book}).
In the standard setting, however, it is typical to view $\Delta$ as a constant
(or as a parameter to be minimized), whereas here we are interested in $\Delta$
as an arbitrary parameter that scales along with $n$.
Accounting for this requires some additional considerations.

\begin{proof}
  Let us first fix some set $S = \{ s_1,\dots,s_{\abs{S}} \}$ of left-vertices
  with $\abs{S} \le \gamma n$ and consider the probability with which the
  property is satisfied conditioned on the choice of $E$.
  Let 
  \[
    Z_{i,j} = \left[ M_i(s_j) \in \bigcup_{k=1}^{i-1} M_k(S) \right],
  \]
  that is, $Z_{i,j} = 1$ if and only if there is a collision between $M_i(s_j)$
  and some $M_k(S)$ for $k < i-1$ (and otherwise $Z_{i,j} = 0$).
  Since the $M_i$ are chosen uniformly at random, we can view the $M_i$ as being
  obtained by picking right-vertices $M_i(s_1),\dots,M_i(s_{\abs{S}})$ in order
  and without replacement.
  Using $\gamma \le 1/2$, this gives us that
  \[
    \Pr[Z_{i,j} = 1] \le \frac{\abs{\bigcup_{k=1}^{i-1} M_k(S)}}{n-\abs{S}}
    \le \frac{\Delta \abs{S}}{n - \abs{S}}
    \le \frac{\gamma \Delta}{1 - \gamma}
    \le 2 \gamma \Delta
  \]
  for every $i \in [\Delta]$ and $j \in [\abs{S}]$.
  The $Z_{i,j}$ are independent for different values of $i$, but for $j < j'$ we
  do \emph{not} have that $Z_{i,j}$ and $Z_{i,j'}$ are independent (since
  $Z_{i,j} = 1$ means $M_i(s_j) \in \bigcup_{k=1}^{i-1} M_k(S)$, which decreases
  the probability of $M_i(s_{j'})$ being picked in the union
  $\bigcup_{k=1}^{i-1} M_k(S)$ since we allow no repeats).
  Nevertheless, observe that $\Pr[Z_{i,j'} = 1 \mid Z_{i,j} = 1] \le
  \Pr[Z_{i,j'} = 1]$ holds, which by induction means we can treat the $Z_{i,j}$
  as independent for the purposes of an upper bound on the probability of their
  \enquote{intersection}; that is, formally,
  \[
    \Pr\left[ \forall k \in [\ell]: Z_{i,j_k} = \ell \right]
    \le \prod_{k=1}^\ell \Pr[ Z_{i,j_k} = 1 ]
  \]
  for any $j_1 < \cdots < j_\ell$.
  Thus using the well-known inequality $\binom{a}{b} \le (ae/b)^b$ for integers
  $a$ and $b$, we get that
  \begin{align*}
    \Pr[\abs{N(S)} < (\Delta - \rho) \abs{S}]
    &\le \Pr\left[\sum_{i,j} Z_{i,j} \ge \rho \abs{S}\right] \\
    &\le \binom{\Delta \abs{S}}{\rho \abs{S}}
      (\max_{i,j} \Pr[Z_{i,j} = 1])^{\rho \abs{S}} \\
    &\le \left( \frac{e\Delta}{\rho} \right)^{\rho\abs{S}}
      \left( 2\gamma\Delta \right)^{\rho\abs{S}} \\
    &\le \left( \frac{2 \gamma e \Delta^2}{\rho} \right)^{\rho \gamma n},
  \end{align*}
  where the last inequality is due to $\abs{S} \le \gamma n$.
  Finally, using a union bound and $\gamma \le 1/2$ once more, the probability
  that $\abs{N(S)} < (\Delta - \rho) \abs{S}$ holds for \emph{any} set of
  left-vertices $S$ (with $\abs{S} \le \gamma n$) is at most
  \begin{align*}
    \sum_{i=1}^{\gamma n} \binom{n}{i} 
      \left( \frac{2 \gamma e \Delta^2}{\rho} \right)^{\rho \gamma n}
    &\le \gamma n \binom{n}{\gamma n}
      \left( \frac{2 \gamma e \Delta^2}{\rho} \right)^{\rho \gamma n} \\
    &\le \gamma n
      \left( \frac{e}{\gamma} \right)^{\gamma n}
      \left( \frac{2 \gamma e \Delta^2}{\rho} \right)^{\rho \gamma n} \\
    &= \gamma n \left(
      e \gamma^{\rho-1} \left( \frac{2e\Delta^2}{\rho} \right)^\rho
    \right)^{\gamma n} \\
    &\le \frac{\gamma n}{2^{\gamma n}} \\
    &< \frac{1}{2}.
    \qedhere
  \end{align*}
\end{proof}

\subsection{Lower Bound for One-sided Error Algorithms}
\label{sec:lb-one-sided}

In this section we restate and prove:

\restateThmLBOneSided*

As mentioned above, the proof requires an appropriate expander construction.
The following lemma shows the two properties we need and how to choose the
parameters of \cref{lem_expander_lower_bound_one_step_construction} in order to
obtain it.
Notice how the expansion guarantee is stronger for right-vertices, which is
offset by the fact that the respective sets that must expand are much smaller.

\begin{lemma}\label{lem_expander_lower_bound_one_step_one_sided_error}
  Let $\Delta = \Delta(n) = \omega(1)$ and $\eps > 0$ be given as functions of
  $n$, and let $\beta = 1/24$.
  Then there is a family $G_n$ of balanced bipartite graphs (with $n$ nodes on
  either side) such that, when $n$ is large enough, $G_n$ has the following
  properties:
  \begin{enumerate}
    \item \emph{Moderate expansion on both sides.}
    For every set $S$ of left- or right-vertices of $G_n$ with $\abs{S} \le
    n/96e\Delta$,
    \[
      \abs{N(S)} \ge (1 - \beta) \Delta \abs{S}.
    \]
    In addition, by \cref{lem_expanders_unique_neighbors}, $S$ has at least
    $(1-2\beta)\Delta\abs{S}$ unique neighbors.
    \item \emph{Small intersections of neighbor sets on the right.}
    Let $\eps = \Omega(\Delta^b / n)$ for some (constant) $b \ge 1$.
    Then for every set $S$ of right-vertices with $\abs{S} \le 2/\eps$ we have
    \[
      \abs{N(S)} \ge (\Delta - r)\abs{S}
    \]
    where $r$ is chosen depending on $b$ as follows:
    \begin{enumerate}
      \item If $b > 2$, then $r$ is constant.
      \item If $b = 2$, then $r = \Theta(\log \Delta)$.
      \item If $1 \le b < 2$, then $r = \Theta(\Delta^{2-b})$.
    \end{enumerate}
    By \cref{lem_expanders_intersection}, in each of these cases we have that,
    for every two disjoint sets of right-vertices $S$ and $S'$ with $\abs{S} \le
    1/\eps$ and $\abs{S'} \ge \abs{S}$,
    \[
      \Pr_{v \in U_{S'}}\left[ \abs*{N(v) \cap N(S)}>2r \right]
        \le \frac{\abs{S}}{\abs{S'}}.
    \]
  \end{enumerate}
\end{lemma}

\begin{proof}
  \begin{enumerate}
    \item We obtain the statement directly by plugging in $\gamma = 1/96e\Delta$
    and $\rho = \beta\Delta$ in
    \cref{lem_expander_lower_bound_one_step_construction}.
    All we need to do is verify that the conditions required in
    \cref{lem_expander_lower_bound_one_step_construction} hold.
    We have $\rho > 1$, $\gamma \ge 2/n$, and $\gamma \le 1/2$ automatically, so
    we only need to show the second upper bound on $\gamma$.
    To see why it holds, let $\Delta$ be large enough so that $2^{\beta\Delta} >
    8e^2\Delta/\beta$.
    Then we have $\beta/2e\Delta > 2e/2^{\beta\Delta-1}$ and in turn
    \[
      \left(
        \frac{1}{2e}
        \left( \frac{\rho}{2 e \Delta^2} \right)^\rho
      \right)^{1/(\rho-1)}
      = \left(
        \frac{1}{2e}
        \left( \frac{\beta}{2 e \Delta} \right)^{\beta \Delta}
      \right)^{1/(\beta\Delta-1)}
      = \left(
        \frac{1}{2e} \cdot \frac{\beta}{2e\Delta}
      \right)^{1/(\beta\Delta-1)} \frac{\beta}{2e\Delta}
      > \frac{\beta}{4e\Delta} 
      = \gamma.
    \]
    \item 
    \begin{enumerate}
      \item Set $r = 1 + 4/(b-2)$ and plug in $\gamma = 2/\eps n$ and $\rho = r$
      in \cref{lem_expander_lower_bound_one_step_construction}.
      As before, we only need to verify that the conditions of
      \cref{lem_expander_lower_bound_one_step_construction} hold.
      Since $b > 2$, we have
      \[
        \Delta^{b/2 + 1}
        = \Delta \cdot \frac{\Delta^b}{\Delta^{b/2}}
        < \Delta \cdot \frac{\Delta^b}{n}
        = o(\eps n).
      \]
      In addition, $r/(r-1) = b/4 + 1/2$, which in turn means that
      \[
        \left( \frac{1}{2e}
          \left( \frac{r}{2e\Delta^2} \right)^r \right)^{1/(r-1)}
        = \left( \frac{1}{2e} \left( \frac{r}{2e} \right)^r \right)^{1/(r-1)}
          \cdot \frac{1}{\Delta^{2r/(r-1)}}
        = \Omega(1) \cdot \frac{1}{\Delta^{b/2+1}}
        = \omega\left( \frac{1}{\eps n} \right).
      \]
      \item Set $r = 1 + 4e^2\log\Delta$.
      Again we plug in $\gamma = 2/\eps n$ and $\rho = r$ in
      \cref{lem_expander_lower_bound_one_step_construction}.
      This is justified since $r/(r-1) = 1+1/4e^2\log\Delta < 1+1/\log\Delta$
      as well as $1/\Delta^2 = \Omega(1/\eps n)$ and then
      \begin{align*}
        \left( \frac{1}{2e}
          \left( \frac{r}{2e\Delta^2} \right)^r \right)^{1/(r-1)}
        &> \left( \frac{1}{2e}
          \left( \frac{2e\log\Delta}{\Delta^2} \right)^r \right)^{1/(r-1)} \\
        &> \left( \frac{\log\Delta}{\Delta^2} \right)^{1+1/\log\Delta} \\
        &> \frac{\log\Delta}{\Delta^2}
          \left( \frac{1}{\Delta^2} \right)^{1/\log\Delta} \\
        &= \frac{\log \Delta}{4\Delta^2} \\
        &= \omega\left( \frac{1}{\eps n} \right).
      \end{align*}
      \item Let $a > 0$ be such that $\eps \ge a \Delta^b / n$.
      Set $r = 1 + 8e^2\alpha\Delta^{2-b}$ where $\alpha = 2\max\{1,1/a\}$.
      Since $\Delta^{2-b} = \omega(\log\Delta)$, we have $\Delta^{2-b} >
      b\log\Delta$ for large enough $\Delta$, and thus $r/(r-1) =
      1+1/8e^2\alpha\Delta^{2-b} < 1+1/\Delta^{2-b} < 1+1/b\log\Delta$.
      Hence this time plugging in $\gamma = 2/\eps n$ and $\rho = r$ in
      \cref{lem_expander_lower_bound_one_step_construction} is justified since
      \[
        \left( \frac{1}{2e}
          \left( \frac{r}{2e\Delta^2} \right)^r \right)^{1/(r-1)}
        > \left( \frac{2\alpha}{\Delta^b} \right)^{1+1/b\log\Delta}
        > \frac{2\alpha}{\Delta^b}
          \left( \frac{1}{\Delta^b} \right)^{1/b\log\Delta}
        = \frac{\alpha}{\Delta^b}
        \ge \frac{2}{a\Delta^b}
        \ge \frac{2}{\eps n}
        = \gamma.
        \qedhere
      \]
    \end{enumerate}
  \end{enumerate}
\end{proof}

With this lemma, we now prove \cref{thm:lb-one-sided}.

\begin{proof}[Proof of \cref{thm:lb-one-sided}.]
  Let $A$ be a non-adaptive, one-sided error tester for $\oBP$ that makes $q$
  queries.
  We use the graph family from
  \cref{lem_expander_lower_bound_one_step_one_sided_error}.

  \paragraph{The distribution.}
  Consider the distribution $D$ of inputs obtained by the following process:
  \begin{enumerate}
    \item Pick a subset $B \subseteq V$ at random by independently adding each
    right-vertex in $V$ to $B$ with probability $6\eps$.
    \item Set $\ENV(v,1) = 0$ for every $v \in V$.
    \item Set $\ENV(v,2) = 1$ if and only if $v \in B$.
  \end{enumerate}
  As we show next, an instance produced by $D$ is $\eps$-far from $\oBP$ with
  high probability.
  (Hence without restriction we may assume $D$ always produces an $\eps$-far
  instance.)
  Intuitively this should be the case since most nodes only have about a $6\eps$
  fraction of their neighbors colored black.
  However, this alone is not enough because it does exclude the possibility of a
  group of nodes $C$ sharing the same white neighbors.
  Indeed, if this is the case, then $N(C)$ can be predominantly black even
  though $N(C)$ contains almost only white nodes.
  This is where the first of the expander guarantees of
  \cref{lem_expander_lower_bound_one_step_one_sided_error} kicks in, forcing $C$
  to have many unique neighbors (and thus $N(C)$ to be predominantly white).

  \paragraph{$\eps$-farness.}
  Since each node is added to $B$ independently with probability $6\eps$, by the
  Chernoff bound (\cref{thm_chernoff}) we have $11\eps n/2 \le \abs{B} \le
  13\eps n/2$ with high probability.
  (Recall we have $n$ left- and $n$ right-vertices.)
  In addition, since $\Delta = \omega(1)$ is the interesting case, we may also
  assume $\Delta/18\eps > \ln(2\Delta/\eps)$.
  Again using that we color each node independently with probability $6\eps$, by
  the Chernoff bound we have that, for any fixed $v \in V$,
  \[
    \Pr\left[ \frac{\abs{N(v) \cap B}}{\abs{N(v)}} \ge \frac{1}{3} \right]
    \le 2e^{-\Delta/18\eps}
    < \frac{\eps}{\Delta}.
  \]
  Let $\ENV'$ be the instance where, for every $v \in V$ that satisfies the
  above, we recolor all of $N(v)$ white; that is, for every $v$ with $\abs{N(v)
  \cap B} \ge \abs{N(v)}/3$, we set $\ENV'(u,2) = 0$ for $u \in N(v)$,
  $\ENV'(u,2) = 0$ if $u \notin N(v)$ for any such $v$, and $\ENV'(u,1) = 0$ for
  every $u \in V$.
  Then certainly $\dist(\ENV,\ENV') < (\eps n/\Delta)\cdot\Delta/4n = \eps/4$.
  By the triangle inequality, it suffices to show $\dist(\ENV',\oBP) \ge
  5\eps/4$.

  Let $\ENV_0 \in \oBP$ be an instance that is at minimum distance to $\ENV'$.
  In addition, let $C = \{ v \in V \mid \ENV_0(v,1) = 1 \}$ and $B_0 = \{ v \in
  V \mid \ENV_0(v,2) = 1 \}$.
  Clearly 
  \[
    \dist(\ENV',\ENV_0) \ge \frac{\abs{B_0 \setminus B}}{4n}
    \ge \frac{\abs{B_0} - \abs{B}}{4n}
    \ge \frac{\abs{B_0} - 13\eps n/2}{4n},
  \]
  so we may assume $\abs{B_0} \le 23\eps n/2$ as otherwise the
  claim is trivial.
  Set $\eps_0 = 1/(13 \cdot 96e) = 1/1248e$, which implies $12\eps n/\Delta <
  n/96e\Delta$.
  Hence we may apply the first item of
  \cref{lem_expander_lower_bound_one_step_one_sided_error} and obtain that, if
  $\abs{C} > 12\eps n/\Delta$, then 
  \[
    \abs{N(C)} > (1-\beta)\Delta \cdot \frac{12\eps n}{\Delta}
    = \frac{23\eps n}{2}
    \ge \abs{B_0}.
  \]
  Thus we must have $\abs{C} \le 12\eps n/\Delta$.
  We argue next that we actually have $C = \varnothing$ and thus also $B_0 =
  \varnothing$.
  Hence $\ENV_0$ must be the instance where $\ENV_0(v,i) = 0$ for every pair
  $(v,i)$, which in particular implies
  \[
    \dist(\ENV,\ENV_0) \ge \frac{\abs{B}}{4n}
    \ge \frac{11\eps}{8}
    > \frac{5\eps}{4}.
  \]

  \cref{lem_expanders_unique_neighbors} tells us that $C$ has at least
  $(1-2\beta)\Delta\abs{C}$ unique neighbors.
  By averaging, there is a node $v \in C$ for which the set $U_v$ of unique
  neighbors of $v$ (i.e., the set of neighbors $u \in N(v)$ for which the only
  node that $u$ is incident to in $C$ is $v$) is such that $\abs{U_v} \ge
  (1-2\beta)\abs{N(v)} = 11\Delta/12$.
  In turn, by the pigeonhole principle (and using $\abs{N(v) \cap B} <
  \abs{N(v)}/3$), this implies that at least $7\Delta/12$ neighbors of $v$ are
  in $U_v \setminus B$.

  We now claim that we can contradict the minimality of $\ENV_0$ by constructing
  $\ENV_v \in \oBP$ with $\dist(\ENV',\ENV_v) < \dist(\ENV',\ENV_0)$.
  Namely we obtain $\ENV_v$ from $\ENV_0$ by coloring $v$ and all its unique
  neighbors white; formally,
  \[
    \ENV_v(u,t) = \begin{cases}
      0, & \text{($u = v$ and $t = 1$) or ($u \in U_v$ and $t = 2$)} \\
      \ENV_0(u,t), & \text{otherwise.}
    \end{cases}
  \]
  Since $\ENV'(v,1) = \ENV_v(v,1) = 0 \neq 1 = \ENV_0(v,1)$ and also
  $\ENV_0(u,2) = 1$ and $\ENV_v(u,2) = 0$ for every $u \in U_v$, we have
  \begin{align*}
    &\dist(\ENV',\ENV_0) - \dist(\ENV',\ENV_v) \\
    &\qquad= \frac{1}{4n} \left(
        \sum_{u,t} \left([\ENV'(u,t) \neq \ENV_0(u,t)] 
          - [\ENV'(u,t) \neq \ENV_v(u,t)]\right)
      \right) \\
    &\qquad= \frac{1}{4n} \Bigg(
        [\ENV'(v,1) \neq \ENV_0(v,1)] - [\ENV'(v,1) \neq \ENV_v(v,1)] \\
    &\qquad\qquad + \sum_{u \in U_v}\left(
          [\ENV'(u,2) \neq \ENV_0(u,2)] - [\ENV'(u,2) \neq \ENV_v(u,2)]\right)
      \Bigg) \\
    &\qquad\ge \frac{1}{4n}
      \left( 1 + \abs{U_v \setminus B} - \abs{U_v \cap B} \right) \\
    &\qquad\ge \frac{1}{4n}
      \left( 1 + \frac{7\Delta}{12} - \frac{\Delta}{3} \right) \\
    &\qquad> 0.
  \end{align*}
  The contradiction arises from the existence of $v$.
  It follows that $C = \varnothing$, as desired.

  \paragraph{Lower bound on $q$ for non-adaptive testers.}
  We first state the argument for non-adaptive testers.
  Once this is in place, it is then simple to obtain the lower bound in the
  adaptive case with a few additional observations.

  Applying Yao's minimax principle, we fix a set $Q \subseteq V$ of $q =
  \abs{Q}$ queried nodes.
  Without restriction, the algorithm $A$ queries all of $\ENV(Q,1)$ and
  $\ENV(Q,2)$ (and so its query complexity is $2q$).
  For $v \in V$ and a set $X \subseteq V$, let
  \[
    \mu_v(X) = \frac{\abs*{N(v) \cap X}}{\abs*{N(v)}}.
  \]
  In addition, for $b \in \binalph$ and $t \in \{1,2\}$, let $Q_{b,t} = \{ v \in
  Q \mid \ENV(v,t) = b \}$.
  Consider the sets
  \[
    C_1 = \left\{ v \in R \mid \mu_v(Q_{0,1}) \ge \frac{1}{2} \right\}
  \]
  and
  \[
    C_2 = \left\{ v \in R \mid \mu_v(N(Q_{0,2})) \ge \frac{1}{2} \right\}.
  \]
  (Recall $R$ is the set of right-vertices.)
  These sets correspond to the two possible strategies for verifying that a node
  $v \in Q_{1,2}$ is incorrectly colored black:
  With $C_1$ one tries to detect if the predecessors of $v$ are white (since, if
  every predecessor of $v$ is white, then $v$ is a type II violation).
  Meanwhile with $C_2$ we are ascertaining the color of the predecessors of $v$
  indirectly by querying their successors; the rationale is that, if we know
  every predecessor of $v$ has a white successor, then we know that (if the
  $\oBP$ rule is being followed correctly) every predecessor of $v$ has to be
  white, and hence $v$ cannot be black.

  In fact it turns out that, in order to reject $\ENV$, $A$ \emph{must} follow a
  mix of these strategies.
  More precisely, $A$ can only reject $\ENV$ if there is at least one $v \in
  Q_{1,2}$ that is also in $C_1 \cup C_2$.
  (This might not be a sufficient condition for $A$ to reject, but it is
  certainly necessary.)
  The reason for this is that otherwise, for every $v \in Q_{1,2}$,
  $\mu_v(Q_{0,1}) + \mu_v(N(Q_{0,2})) < 1$ and thus we always have some
  predecessor $p_v \in N(v) \setminus (Q_{0,1} \cup N(Q_{0,2}))$.
  We claim that, in this case, we can obtain $\ENV' \in \oBP$ with $\ENV(Q,1) =
  \ENV'(Q,1)$ and $\ENV(Q,2) = \ENV'(Q,2)$ as follows:
  \[
    \ENV'(u,1) = \begin{cases}
      1, & \exists v \in Q_{1,2}: u = p_v \\
      0, & \text{otherwise}
    \end{cases}
  \]
  and
  \[
    \ENV'(u,2) = \begin{cases}
      1, & \exists v \in Q_{1,2}: u \in N(p_v) \\
      0, & \text{otherwise.}
    \end{cases}
  \]
  Since $\ENV' \in \oBP$ is evident (as $v \in N(p_v)$), we need only check that
  $\ENV(u,t) = \ENV'(u,t)$ for every $u \in Q$:
  \begin{itemize}
    \item If $t = 1$, then $\ENV(u,t) = 0$.
    The only case where $\ENV'(u,t) \neq 0$ is when $u = p_v$ for some $v \in
    Q_{1,2}$, and thus by definition $u \notin Q$.
    \item If $\ENV(u,2) = 0$, then $u \in Q_{0,2}$ and so there is no $v \in
    Q_{1,2}$ for which $p_v \in N(u)$.
    Hence we also have $\ENV'(u,2) = 0$.
    \item Finally, if $\ENV(u,2) = 1$, then $u \in Q_{1,2}$ and so $u \in
    N(p_u)$ and consequently $\ENV'(u,2) = 1$.
  \end{itemize}

  Moreover, observe that, for $A$ to succeed, at least one of $\abs{C_1 \cap
  Q_{1,2}}$ and $\abs{C_2 \cap Q_{1,2}}$ must be at least $1/26\eps$.
  Otherwise, by the union bound, the probability that $B \cap (C_1 \cup C_2)$ is
  non-empty is at most $6/13 < 1/2$ (and, when $B \cap (C_1 \cup C_2)$ is empty,
  certainly there is no $v \in Q_{1,2}$ that would cause $A$ to reject).
  We now derive the lower bound from either of the two possibilities $\abs{C_1
  \cap Q_{1,2}} \ge 1/26\eps$ and $\abs{C_2 \cap Q_{1,2}} \ge 1/26\eps$ as
  follows:
  \begin{description}
    \item[$\abs{C_1 \cap Q_{1,2}} \ge 1/26\eps$.] 
    Consider an arbitrary set $C_1' \subseteq C_1 \cap Q_{1,2}$ with $\abs{C_1'}
    = \min\{ 1/26\eps, n/96e\Delta \}$.
    Since $\eps = \Omega(\Delta/n)$, we have that $\abs{C_1'} = \Omega(1/\eps)$.
    For $v \in C_1'$, let $U_v$ denote the set of unique neighbors of $v$.
    By the first item of
    \cref{lem_expander_lower_bound_one_step_one_sided_error}, $C_1'$ has at
    least $(1-2\beta)\Delta\abs{C_1'}$ unique neighbors.
    Hence, if we draw $v$ uniformly at random from $C_1'$, then $\abs{N(v)
    \setminus U_v} \le 2\beta\Delta$ in expectation.
    By Markov's inequality, this implies at least $\abs{C_1'}/2$ many nodes $v
    \in C_1'$ are such that $\abs{N(v) \setminus U_v} \le 4\beta\Delta$. 
    For each such node $v$, since we have $\mu_v(Q_{0,1}) \ge 1/2$, in
    particular $A$ must query at least half the nodes in $N(v)$, so at least
    $(1/2-4\beta)\Delta = \Delta/3$ neighbors that are unique to $v$.
    Hence $q \ge (\Delta/3) \cdot \abs{C_1'}/2 = \Omega(\Delta/\eps)$. 

    \item[$\abs{C_2 \cap Q_{1,2}} \ge 1/26\eps$.] 
    Let $C_2' \subseteq C_2 \cap Q_{1,2}$ be such that $\abs{C_2'} = \min\{
    1/26\eps, n/96e\Delta \}$.
    Without restriction, we assume that $\abs{Q_{0,2}} \ge \abs{C_2'}$.
    Arguing as before, at least $\abs{C_2'}/2$ many nodes $v \in C_2'$ are such
    that $\abs{U_v} \ge (1-4\beta)\Delta$. 
    Since each $v \in C_2'$ must be such that at least half of its neighbors are
    covered by at least one node in $Q_{0,2}$, the set $Q_{0,2}$ must thus cover
    at least $(1/2-4\beta)\Delta\abs{C_2'}/2 = \Delta\abs{C_2'}/6$ nodes of
    $\bigcup_{v \in C_2'} U_v$.
    As $C_2' \subseteq Q_{1,2}$ and $Q_{0,2}$ are disjoint by definition, we
    plug in $S = C_2'$ and $S' = Q_{0,2}$ in the second item of
    \cref{lem_expander_lower_bound_one_step_one_sided_error} to obtain that
    there is some set $Z \subseteq Q_{0,2}$ with $\abs{Z} \le \abs{C_2'}$ such
    that, for every $v \in Q_{0,2} \setminus Z$,
    \[
      \abs{N(v) \cap N(C_2')} \le 2r,
    \]
    where $r$ is as in \cref{lem_expander_lower_bound_one_step_one_sided_error}.
    Furthermore, using the first item of \cref{lem_expanders_intersection}
    together with the expansion guarantee of
    \cref{lem_expander_lower_bound_one_step_one_sided_error}, we have that
    $\abs{N(Z) \cap N(C_2')} \le 2r\abs{C_2'}$.
    Hence, excluding the nodes in $N(Z)$, we are left with
    $(\Delta/6-2r)\abs{C_2'} = \Omega(\Delta/\eps)$ nodes in $N(C_2')$ to be
    covered by $Q_{0,2} \setminus Z$, of which each node can contribute with at
    most $2r$ covered nodes.
    It follows that $q = \Omega(\Delta/\eps r)$.
    A case-by-case analysis based on the values of $b$ and $r$ as given by the
    second item in \cref{lem_expander_lower_bound_one_step_one_sided_error}
    yields the bounds in the theorem statement.
  \end{description}

  \paragraph{Lower bound on $q$ for adaptive testers.}
  We now use the above to derive the bounds for the case where $A$ is adaptive.
  (Recall that we are still relying on Yao's minimax principle, so $A$ is
  deterministic.)
  Notice first that our distribution $D$ generates instances that are all-white
  in the first step and have $O(\eps n)$ black nodes in the second one.
  Hence $A$ only sees white nodes (with, say, at least $4/5$ probability) unless
  it makes $\Omega(1/\eps)$ second-step queries.
  This gives the first term in the lower bounds.

  Continuing our adaptation of the argument, we define $Q$, $Q_{b,t}$, and $C_i$
  as before, which are now random variables conditioned on $D$ and the value of
  which is determined by the logic of $A$.
  (That is, to sample any one of the mentioned sets, first produce an
  environment $\ENV$ according to $D$; then run $A$ on $\ENV$ to obtain the set
  of queries $Q$, which in turn uniquely defines the other sets $Q_{b,t}$ and
  $C_i$.)
  Despite these now being random variables, we still have that $A$ can only
  reject the input if $Q_{1,2} \cap (C_1 \cup C_2)$ is non-empty.
  Otherwise we can construct $\ENV' \in \oBP$ as before and, since $A$ is a
  one-sided error algorithm and its view of $\ENV$ and $\ENV'$ is the same, it
  cannot reject $\ENV$ without also rejecting $\ENV'$.
  The lower bound is now weaker since we can no longer derive that $\abs{C_1
  \cap Q_{1,2}}$ or $\abs{C_2 \cap Q_{1,2}}$ must be large for $A$ to accept;
  rather, all that we can say is that these sets are non-empty.
  Nevertheless, the argument still goes through if we assume only $\abs{C_1'} =
  1$ or $\abs{C_2'} = 1$ and thus we obtain a lower bound of
  $\Omega(\min\{\Delta,\Delta/r\}) = \Omega(\Delta/r)$ on the number of queries.
\end{proof}

\subsection{Lower Bound for Two-sided Error Algorithms}
\label{sec:lb-two-sided}

In this section we prove the lower bound for two-sided error algorithms:

\restateThmLBTwoSided*

As before, we use \cref{lem_expander_lower_bound_one_step_construction} to
construct the graphs that we need for the proof of \cref{thm:lb-two-sided}.

\begin{lemma}\label{lem_expander_lower_bound_one_step}
  Let $\Delta = \Delta(n) = \omega(1)$ and $\eps = \Omega(\Delta^b/n)$ be given
  as a function of $n$, where $b \ge 1$ is some constant.
  In addition, let $\beta = 1/48^2$ and $\eps \le \beta/192e$.
  Then there is a family $G_n$ of balanced bipartite graphs with $n$ nodes on
  either side such that, when $n$ is large enough, $G_n$ has the following
  properties:
  \begin{enumerate}
    \item \emph{Moderate expansion on the left.}
    For every set $S$ of left-vertices of $G_n$ with $\abs{S} \le 2 \alpha n /
    \Delta$ where $\alpha = 24\eps$,
    \[
      \abs{N(S)} \ge (1 - \beta) \Delta \abs{S}.
    \]
    \item \emph{Unique-neighbor expansion on the right.}
    For every set $S$ of right-vertices of $G_n$ with $\abs{S} \le \min\{
    \Delta/\eps, \sqrt{n/\eps\Delta} \}$, we have
    \[
      \abs{N(S)} \ge (\Delta - r) \abs{S}
    \]
    where the value of $r$ depends on $b$ as follows:
    \begin{enumerate}
      \item If $b > 3$, then $r$ is constant.
      \item If $b = 3$, then $r = \Theta(\log\Delta)$.
      \item If $1 \le b < 3$, then $r = \Theta(\Delta^{(3-b)/2})$.
    \end{enumerate}
    By \cref{lem_expanders_unique_neighbors}, this implies $S$ has at least
    $(\Delta-2r)\abs{S}$ unique neighbors.
  \end{enumerate}
\end{lemma}

\begin{proof}
  \begin{enumerate}
    \item Proceeding as in the proof of
    \cref{lem_expander_lower_bound_one_step_one_sided_error}, we plug in $\gamma
    = 2\alpha / \Delta$ and $\rho = \beta \Delta$ in
    \cref{lem_expander_lower_bound_one_step_construction}.
    Since $\Delta = \omega(1)$, for sufficiently large $n$ we have
    $2^{\beta\Delta} > 8e^2\Delta/\beta$, thus implying $\beta/2e\Delta >
    2e/2^{\beta\Delta-1}$.
    In turn this means that
    \[
      \left(
        \frac{1}{2e}
        \left( \frac{\rho}{2 e \Delta^2} \right)^\rho
      \right)^{1/(\rho-1)}
      = \left(
        \frac{1}{2e}
        \left( \frac{\beta}{2 e \Delta} \right)^{\beta \Delta}
      \right)^{1/(\beta\Delta-1)}
      = \left(
        \frac{1}{2e} \cdot \frac{\beta}{2e\Delta}
      \right)^{1/(\beta\Delta-1)} \frac{\beta}{2e\Delta}
      > \frac{\beta}{4e\Delta} 
      \ge \gamma.
    \]
    \item We plug in $\gamma = \abs{S} / n$ and $\rho = r$ in
    \cref{lem_expander_lower_bound_one_step_construction}. 
    Let us consider the three cases separately:
    \begin{enumerate}
      \item If $b > 3$, we set $r = 1 + 4/(b-3)$.
      Then we have $2r/(r-1) = (b+1)/2$ as well as $(b+1)/2 < b-1$ and $\gamma
      \le \Delta/\eps n = O(1/\Delta^{b-1})$.
      Putting all of this together, we get
      \[
        \left( \frac{1}{2e}
          \left( \frac{r}{2e\Delta^2} \right)^r \right)^{1/(r-1)}
        = \left( \frac{1}{2e} \left( \frac{r}{2e} \right)^r \right)^{1/(r-1)}
          \cdot \frac{1}{\Delta^{2r/(r-1)}}
        = \Omega(1) \cdot \frac{1}{\Delta^{(b+1)/2}}
        = \omega(\gamma).
      \]
      \item Let $b = 3$, in which case $\gamma = O(1/\Delta^2)$.
      Setting $r = 1 + 4e^2\log\Delta$, we have $r/(r-1) = 1 + 1/4e^2\log\Delta
      < 1 + 1/\log \Delta$ and then
      \begin{align*}
        \left( \frac{1}{2e}
          \left( \frac{r}{2e\Delta^2} \right)^r \right)^{1/(r-1)}
        &> \left( \frac{1}{2e}
          \left( \frac{2e\log\Delta}{\Delta^2} \right)^r \right)^{1/(r-1)} \\
        &> \left( \frac{\log\Delta}{\Delta^2} \right)^{1+1/\log\Delta} \\
        &> \frac{\log\Delta}{\Delta^2}
          \left( \frac{1}{\Delta^2} \right)^{1/\log\Delta} \\
        &= \frac{\log \Delta}{4\Delta^2} \\
        &= \omega(\gamma).
      \end{align*}
      \item Let $b < 3$ and let $a > 0$ be such that $\eps \ge a\Delta^b/n$.
      Note this implies $\gamma \le 1/\sqrt{a \Delta^{b+1}}$ as well as
      $\Delta^{(3-b)/2} = \omega(\log\Delta)$.
      Set $r = 1 + 8e^2\alpha\Delta^{(3-b)/2}$ where $\alpha =
      \max\{1,1/\sqrt{a}\}$.
      Then for large values of $\Delta$ we have $r/(r-1) = 1 +
      1/8e^2\alpha\Delta^{(3-b)/2} < 1 + 2/(b+1)\log\Delta$.
      With these observations, we get that
      \begin{align*}
        \left( \frac{1}{2e}
          \left( \frac{r}{2e\Delta^2} \right)^r \right)^{1/(r-1)}
        &> \left(
          \frac{2\alpha}{\sqrt{\Delta^{b+1}}}
        \right)^{1+2/(b+1)\log\Delta} \\
        &> \frac{2\alpha}{\sqrt{\Delta^{b+1}}}
          \left( \frac{1}{\sqrt{\Delta^{b+1}}} \right)^{2/(b+1)\log\Delta} \\
        &= \frac{\alpha}{\sqrt{\Delta^{b+1}}} \\
        &\ge \frac{1}{\sqrt{a\Delta^{b+1}}} \\
        &\ge \gamma.
        \qedhere
      \end{align*}
    \end{enumerate}
  \end{enumerate}
\end{proof}

We are now in position to prove \cref{thm:lb-two-sided}.

\begin{proof}[Proof of \cref{thm:lb-two-sided}]
  Let $G = (V,E)$ and $\alpha = 24\eps$ be as in
  \cref{lem_expander_lower_bound_one_step}.
  We assume $\Delta = o(n)$ since that is the interesting case.
  In addition, let $c > 0$ be a constant so that $r \le c$, $r \le c\log\Delta$,
  or $r \le c\Delta^{(3-b)/2}$, depending on which case of $b$ we are
  considering, as per \cref{lem_expander_lower_bound_one_step}.
  Using Yao's minimax principle, we define two distributions $D_Y$ and $D_N$
  such that $D_Y \in \oBP$ and $D_N$ is $\eps$-far from $\oBP$ with high
  probability, and then show that the two are indistinguishable if we can
  observe a single fixed set of vertices $Q$ of $G$ where:
  \begin{itemize}
    \item If $b > 3$, then $\abs{Q} \le \Delta/576c\eps$.
    \item If $b = 3$, then $\abs{Q} \le \Delta/576c\eps\log\Delta$.
    \item If $1 \le b < 3$, then $\abs{Q} \le \Delta^{(b-1)/2}/576c\eps$.
  \end{itemize}

  \paragraph{The distributions.}
  We define $D_Y$ and $D_N$ as follows:
  \begin{itemize}
    \item[$D_Y$:] Pick a set $S \subseteq V$ of left-vertices by randomly and
    independently adding each vertex to $S$ with probability $\alpha /
    3\Delta$.
    Color $S$ black in both steps by setting $\ENV(S,1) = \ENV(S,2) = 1$.
    In addition, color all of $N(S)$ black in the second step; that is,
    $\ENV(N(S),2) = 1$.
    Set $\ENV(x,y) = 0$ for every other pair $(x,y)$ that was not assigned.
    \item[$D_N$:] Pick a set $S \subseteq V$ of left-vertices by randomly and
    independently adding each vertex to $S$ with probability $\alpha /
    \Delta$.
    (We leave $S$ white.)
    For each node $v \in S$, select each right-vertex in $N(v)$ with
    probability $1/3$ uniformly at random and add it to a set $R_v$.
    Letting $B = \bigcup_{v \in S} R_v$ be the union of these right-vertices,
    color them all black in the second step by setting $\ENV(B,2) = 1$.
    Set $\ENV(x,y) = 0$ for every other pair $(x,y)$ that was not assigned.
  \end{itemize}
  Clearly $D_Y$ is always in $\oBP$.
  We next show that $\ENV$ as produced by $D_N$ is $\eps$-far from $\oBP$ with
  high constant probability.
  Intuitively this should hold because the vertices in $S$ have each roughly
  only a $1/3$ fraction of their neighbors colored black.
  As in the proof of \cref{thm:lb-one-sided}, however, this alone does not
  suffice and we need the expansion guarantees to ensure the argument goes
  through.
  
  \paragraph{$\eps$-farness of $D_N$.}
  Let $U$ be the set of unique neighbors of $S$. 
  We first show that the following facts hold with arbitrarily high constant
  probability (e.g., at least $0.99$):
  \begin{enumerate}
    \item $(\alpha-\eps)n/\Delta \le \abs{S} \le (\alpha+\eps)n/\Delta$
    \item $(1/3-\beta)\abs{U} < \abs{B \cap U} < (1/3+\beta)\abs{U}$
    \item $7 \eps n < \abs{B} < 9 \eps n$
  \end{enumerate}

  The first fact is simple to obtain:
  By the Chernoff bound (\cref{thm_chernoff}), when picking each left-vertex
  with probability $\alpha / \Delta = 24\eps/\Delta$, the probability that
  $\abs{S} / n$ deviates from $\alpha / \Delta$ by a factor larger than
  $\eps/\Delta$ is at most $2e^{-n\eps/72\Delta}$. 
  This probability can be made arbitrarily small by using $\eps \ge
  \zeta\Delta/n$ and choosing $\zeta$ appropriately. 

  To show the second fact, we start with
  \cref{lem_expander_lower_bound_one_step,lem_expanders_unique_neighbors},
  which show that (conditioned on our first fact) we have $\abs{U} \ge
  (1-2\beta)\Delta\abs{S}$.
  Recall we sample $R_v$ independently for each $v \in S$; hence the
  probability that we put any fixed $u \in U$ in $B$ is exactly $1/3$.
  Then, by the Chernoff bound, 
  \[
    \Pr\left[
      \abs*{\abs{B \cap U} - \frac{\abs{U}}{3}} > \beta\abs{U}
    \right]
    < 2e^{-\beta^2\abs{U}}
    = o(1)
  \]
  since $\abs{U} = \Omega(\abs{S})$, $\abs{S} = \Omega(n/\Delta)$, and $\Delta
  = o(1)$.
  
  Finally, for the last fact, first notice that the observations above already
  imply
  \[
    \abs{B} \ge \abs{B \cap U}
    \ge \left( \frac{1}{3} - \beta \right)\left(1 - 2\beta\right)\Delta\abs{S}
    > 7\eps n.
  \]
  To obtain the upper bound on $\abs{B}$, recall we have $\abs{N(S) \setminus
  U} \le 2\beta\Delta\abs{S}$.
  Thus with high probability we have
  \[
    \abs{B} = \abs{B \cap U} + \abs{B \setminus U}
    \le \left( \frac{1}{3} + \beta \right) \abs{U}
      + \abs{N(S) \setminus U}
    \le \left( \frac{1}{3} + 3\beta \right) \Delta \abs{S}
    < 9 \eps n.
  \]

  We now show that, assuming the facts above, $\dist(\ENV,\ENV') \ge 4\eps$
  holds for every $\ENV' \in \oBP$.
  Notice that $\ENV' \in \oBP$ is uniquely determined by a set $S'$ of nodes
  for which we set $\ENV'(v,1) = \ENV'(N(v),2) = 1$ for every $v \in S'$ and
  $\ENV'(x,y) = 0$ for every other pair $(x,y)$.
  First notice that we may assume $\abs{S'} \le 15 \eps n/\Delta$ since
  otherwise, by \cref{lem_expander_lower_bound_one_step}, $\abs{N(S')} \ge
  15(1-\beta)\eps n > \abs{B} + 4\eps n$ (and thus $\dist(\ENV,\ENV') \ge
  \abs{N(S') \setminus B} \ge 4\eps n$).

  We handle $S_1' = S' \cap S$ and $S_2' = S' \setminus S$ separately.
  By the first item of \cref{lem_expanders_intersection}, there are at most
  $\beta\Delta(\abs{S} + \abs{S'}) \le 38\beta\eps n < \eps n$ nodes in
  $N(S_2') \cap B$.
  Hence we are done unless $\abs{N(S_1') \cap B} \ge \abs{B} - 5\eps n > 2\eps
  n$, meaning $\abs{S_1'} > 2\eps n/\Delta$.
  Let $U'$ be the set of unique neighbors of $S_1'$.
  Using the same kind of analysis as before, we obtain $\abs{U'} \ge
  (1-2\beta)\Delta\abs{S_1'}$ and $\abs{N(S_1') \setminus U'} \le
  2\beta\Delta\abs{S_1'}$.
  In addition we have $\abs{B \cap U'} \le (1/3 + \beta)\abs{U'}$ with
  arbitrarily high constant probability (since $\abs{S_1'} =
  \Omega(\eps n/\Delta)$).
  Thus with arbitrarily high constant probability
  \[
    \abs{N(S_1') \cap B} \le \abs{B \cap U'} + \abs{N(S_1') \setminus U'}
    < \left( \frac{1}{3} + 3\beta \right) \Delta\abs{S_1'}
    < \frac{2}{5} \Delta\abs{S_1'}.
  \]
  It follows that the nodes in $S_1'$ are actually \emph{increasing}
  $\dist(\ENV,\ENV')$.
  Hence $\dist(\ENV,\ENV')$ is at least as large as in the case $\abs{S_1'}
  \le 2\eps n/\Delta$, in which we already had $\dist(\ENV,\ENV') \ge 4\eps
  n$.
  
  \paragraph{Indistinguishability of $D_Y$ and $D_N$.}
  Fix a set $Q = \{ u_1,\dots,u_q \}$ of vertices of $G$.
  We show that the total variation distance between $D_Y$ and $D_N$ is
  (strictly) less than $1/6$ if we restrict our view to $Q$.

  First notice that $\Pr_{D_Y}[S \cap Q = \varnothing] \ge \Pr_{D_N}[S \cap Q
  = \varnothing]$ and also
  \[
    \Pr_{D_N}[S \cap Q = \varnothing]
    \ge \left( 1 - \frac{\alpha}{\Delta} \right)^q
    \ge 1 - \frac{\alpha q}{\Delta}
    > \frac{11}{12}
  \]
  since we independently put each node in $S$ with probability $\alpha/\Delta$
  and $q < \Delta/576\eps$.
  Hence we may safely assume that $Q$ contains only \emph{right-vertices}; if
  $Q$ contains any left-vertices, then with probability at least $11/12$ these
  vertices are all white in both distributions (which gives the tester no
  advantage in distinguishing $D_Y$ from $D_N$).

  Next we argue that, with high probability over both $D_Y$ and $D_N$, no two
  distinct nodes in $Q$ share a common neighbor in $S$; that is, there are no
  two nodes $u_i \neq u_j$ for which $N(u_i) \cap N(u_j) \cap S$ is non-empty.
  Observe that $q \le \min\{ \Delta/\eps, \sqrt{n/\eps\Delta} \}$ and $2\alpha
  rq/\Delta = 48\eps rq/\Delta \le 1/12$ holds in all three cases:
  \begin{itemize}
    \item If $b > 3$, then $q \le \Delta/576c\eps$.
    Hence certainly we have not only $q \le \Delta/\eps$ but also $\Delta/\eps
    = O(\sqrt{n/\eps\Delta^{b-2}}) = o(\sqrt{n/\eps\Delta})$.
    Since $r \le c$, we also have $48\eps rq/\Delta \le r/12c \le 1/12$.
    \item If $b = 3$, then $\abs{Q} \le \Delta/576c\eps\log\Delta$.
    Therefore obviously $q = o(\Delta/\eps)$ and also $q =
    o(\sqrt{n/\eps\Delta})$ since $\Delta/\eps = O(\sqrt{n/\eps\Delta^{b-2}})
    = O(\sqrt{n/\eps\Delta})$.
    Since $r \le c\log\Delta$, we also have $48\eps rq/\Delta \le
    r/12c\log\Delta \le 1/12$.
    \item If $1 \le b < 3$, then $\abs{Q} \le \Delta^{(b-1)/2}/576c\eps$.
    Thus $q = o(\Delta/\eps)$ and also
    \[
      q < \frac{\Delta^{(b-1)/2}}{\eps}
      \le \Delta^{(b-1)/2} \sqrt{\frac{n}{\eps\Delta^b}}
      = \sqrt{\frac{n}{\eps\Delta}}.
    \]
    Furthermore, $r \le c\Delta^{(3-b)/2}$ implies $48\eps rq/\Delta \le
    r/12c\Delta^{(3-b)/2} \le 1/12$.
  \end{itemize}
  We apply \cref{lem_expander_lower_bound_one_step} and obtain that $Q$ has at
  most $2rq$ many neighbors that are not unique.
  Letting $E$ be the set of such neighbors, we have $\Pr_{D_Y}[S \cap E =
  \varnothing] \ge \Pr_{D_N}[S \cap E = \varnothing]$ and also
  \[
    \Pr_{D_N}[S \cap E = \varnothing]
    = \left( 1 - \frac{\alpha}{\Delta} \right)^{2rq}
    \ge 1 - \frac{2\alpha rq}{\Delta}
    \ge \frac{11}{12}.
  \]
  Hence we may assume that this also holds.

  Now let $Z_i^Y$ and $Z_i^N$ be the indicator function of $q_i$ being colored
  black in $D_Y$ and $D_N$, respectively.
  Since any two nodes $q_i \neq q_j$ of $Q$ do not share a common neighbor in
  $S$ and a node being in $S$ or not is independent from any other node being
  in $S$, $Z_i^X$ and $Z_j^X$ are independent for each of $X \in \{ Y,N \}$.
  In addition, $Z_i^Y$ and $Z_i^N$ are identically distributed because the 
  probability that $q_i$ is colored black on account of some $v \in N(q_i)$ in
  $D_N$ is
  \[
    \Pr[v \in S \land q_i \in R_v] = \frac{\alpha}{3\Delta},
  \]
  which is the same probability that $v \in S$ in $D_Y$ (and thus $q_i$ is
  colored black on account of $v$ in $D_Y$).
  Therefore $D_Y$ and $D_N$ are indistinguishable if we look only at the
  vertices in $Q$ provided the two assumptions we made before hold, which is
  the case with probability at least $5/6$.
\end{proof}


\section{Upper Bounds for the Case \texorpdfstring{$T > 2$}{T > 2}}
\label{sec:alg-gen-T}

In the previous sections we focused on regimes where $T = 2$.
In this section we consider two different strategies for the case where $T > 2$.

An immediate observation to make is that the diameter $\diam(G)$ plays a much
more significant role in this setting.
For instance, the case where $T \ge (1+2/\eps) \diam(G)$ is more or less trivial
since then after $\diam(G)$ steps every connected component must be either
all-black or all-white and the first $\diam(G)$ steps constitute at most an
$\eps/2$ fraction of $\ENV$.

\subsection{Structure-independent Upper Bound}

First we give a generalization of \cref{thm:alg-low-degree}, which is simple to
obtain and is still adequate for settings where $\Delta$ and $T$ are not too
large.
For constant $\Delta$, for instance, it still gives a sublinear query algorithm
whenever $T = O(\log n)$.
As the algorithm of \cref{thm:alg-low-degree}, the algorithm does not use the
graph structure in any way except for knowing what is the neighborhood of each
node.

\restateThmAlgLowDegreeGeneralT*

We adapt \cref{alg:m2} to obtain \cref{alg:m2_gen}.

The analysis does not carry over automatically since we need to consider what
happens if we are correcting violations in a time step $t < T$.
Unlike in \cref{lem:meq2_viol_vs_dist}, this kind of correction may now
propagate to time steps after $t$.
In addition, we have to assume $\Delta \ge 2$; however, the case $\Delta = 1$ is
trivial since then $\diam(G) = 1$ and we need only follow the strategy described
at the beginning of this section.

\begin{algorithm}
  Pick $Q \subseteq V \times \{ t \mid 2 \le t \le T \}$ uniformly at random
  where $\abs{Q} = \ceil{2\Delta^{T-2} / \eps T}$\;
  Query $\ENV(v,t-1)$ and $\ENV(u,t)$ for every $(u,t) \in Q$ and $v \in N(u)$
  in a time-conforming manner\;
  \For{$(u,t) \in Q$}{
    \lIf{$\ENV(u,t) = 0$ and $\exists v \in N(u): \ENV(v,t-1) = 1$}{
      \Reject
    }
    \lIf{$\ENV(u,t) = 1$ and $\forall v \in N(u): \ENV(v,t-1) = 0$}{
      \Reject
    }
  }
  \Accept\;
  \caption{Structure-independent algorithm for the case of general $T$ with
  query complexity $O(\Delta^{T-1}/\eps)$}
  \label{alg:m2_gen}
\end{algorithm}

\begin{lemma}%
  \label{lem:gen-T-relation-viol-dist}
  Let $\Delta \ge 2$.
  Then
  \[
    \frac{\abs{\viol(\ENV)}}{(\Delta+1) nT}
    \le \dist(\ENV, \oBP)
    \le \frac{\Delta^{T-1} - 1}{(\Delta - 1)nT} \abs{\viol(\ENV)}.
  \]
\end{lemma}

\begin{proof}
  The lower bound is as in \cref{lem:meq2_viol_vs_dist} except that every color
  flip may now correct at most $\Delta + 1$ violating pairs (instead of just
  $\Delta$ many pairs).
  The $+1$ is due to the fact that flipping the color of a pair $(u,t)$ for $2
  \le t < T$ may not only correct pairs $(v,t+1)$ where $v \in N(u)$ but also
  $(u,t)$ itself, which was impossible in the setting of
  \cref{lem:meq2_viol_vs_dist}.

  As for the upper bound, the point is that, if we wish to correct $(u,t) \in
  \viol(\ENV)$ by flipping its color, then in the worst case we must flip every
  $(v,t')$ where $t' > t$ and $\dist(u,v) = t' - t$.
  The number of such pairs is at most the number of nodes in a complete
  $\Delta$-ary tree of height $T - 2$, which is $(\Delta^{T-1} - 1)/(\Delta -
  1)$.
\end{proof}

Now as before with \cref{thm:alg-low-degree} we have that $\dist(\ENV,\oBP) \ge
\eps$ implies
\[
  \viol(\ENV) \ge \frac{\eps(\Delta-1)nT}{\Delta^{T-1}-1}
  > \frac{\eps nT}{2\Delta^{T-2}}.
\]
Hence the probability that \cref{alg:m2_gen} errs is
\[
  \Pr[Q \cap \viol(\ENV) = \varnothing]
  \le \left( 1-\frac{\eps T}{2\Delta^{T-2}} \right)^{\abs{Q}}
  < \frac{1}{e} < \frac{1}{2}.
\]
As was the case with \cref{alg:m2}, the query complexity and other properties
required in \cref{thm:alg-low-degree-gen-T} are clear, and hence
\cref{thm:alg-low-degree-gen-T} follows.

\subsection{Upper Bound Based on Graph Decompositions}
\label{sec:alg-low-diam}

The second algorithm we present is suited for not too small values of $T$ and
graphs that are not too dense.

\restateThmAlgLowDiam*

The strategy followed by the algorithm relies on graph decompositions.
These are partitions induced by sets of edges that cut the graph into components
of bounded diameter.

\begin{definition}%
  \label{def:decomposition}
  Let $d \in \N_+$ and $\alpha > 0$.
  A \emph{$(d,\alpha)$-decomposition} of a graph $G = (V,E)$ is a set of edges
  $C \subseteq E$ with $\abs{C} \le \alpha\abs{E}$ and such that there is a
  partition $V = V_1 + \cdots + V_r$ satisfying the following:
  \begin{enumerate}
    \item For $u,v \in V$, $uv \in C$ if and only if there are $i$ and $j$ with
    $i \neq j$ such that $u \in V_i$ and $v \in V_j$.
    \item For every $i$, $\diam(V_i) \le d$.
  \end{enumerate}
\end{definition}

The following is a renowned result in graph decompositions:

\begin{theorem}[\cite{bartal96_probabilistic_focs}]%
  \label{thm:decomposition-gen}
  For any $d \in \N_+$, every graph $G$ admits a
  $(d,O(\log(n)/d))$-decomposition.
\end{theorem}

This trade-off is optimal for graphs in general.
For the special case of graphs excluding a fixed minor (which includes most
notably planar graphs or also graphs of bounded genus), we have the following
small improvement:

\begin{theorem}[\cite{klein93_excluded_stoc}]
  \label{thm:decomposition-planar}
  Let $H$ be a fixed graph.
  For any $d \in \N_+$, every graph $G$ excluding $H$ as a minor admits a
  $(d,O(1/d))$-decomposition.
\end{theorem}

The claim is that \cref{alg:low-diam} satisfies the requirements of
\cref{thm:alg-low-diam}.
As mentioned in the introduction, the strategy followed by the algorithm is
loosely based on a similar testing routine from the paper by
\textcite{nakar21_back_icalp}.
In a nutshell, the idea is to split the environment into more manageable
components and then use the properties of the local rule to predict how each
component must behave.

\begin{algorithm}
  $t_1 \gets \floor{\eps T/4}$\;
  Compute a $(t_1,\alpha)$-decomposition of $G$ according to
  \cref{thm:decomposition-gen} or \cref{thm:decomposition-planar} and obtain a
  set of edges $C$ that cuts $G$ into components $V_1,\dots,V_r$ as in
  \cref{def:decomposition}\;
  $B \gets \{ v \mid \text{$v$ is incident to an edge in $C$} \}$\;
  Pick $Q \subseteq \{ (v,t) \mid \text{$v \in V_i$ and $t \ge t_1$} \}$
  uniformly at random where $\abs{Q} = \ceil{3/\eps}$\;
  $Q' \gets \{ v \in V \mid \exists t: (v,t) \in Q \}$\;
  Query $\ENV(B,t_1)$, $\ENV(Q)$, and $\ENV(Q',t_1)$ in a time-conforming
  fashion\;
  \lIf{$\ENV(B,t_1)$ is not feasible}{\Reject}
  \For{$i \in [r]$}{
    $B_i \gets B \cap V_i$\;
    $B_i' \gets \{ u \in B_i \mid \ENV(u,t_1) = 1 \}$\;
  }
  \For{$v \in V$}{
    \For{$i \in [r]$}{
      \eIf{$B_i' \neq \varnothing$}{
        $\alpha_i(v) \gets 
          \min_{u \in B_i' \cup (V_i \setminus B_i)} \dist(u,v)$\;
        $\beta_i(v) \gets \min_{u \in B_i'} \dist(u,v)$\;
      }{
        $\alpha_i(v) \gets \infty$\;
        $\beta_i(v) \gets \infty$\;
      }
    }
    $\alpha(v) \gets \min_i \alpha_i(v)$\;
    $\beta(v) \gets \min_i \beta_i(v)$\;
  }
  \For{$(v,t) \in Q$}{
    Let $i$ be such that $v \in V_i$\;
    \eIf{$\ENV(v,t_1) = 1$}{
      \lIf{$\ENV(v,t) \neq 1$}{\Reject}
    }{
      \lIf{$t_1 \le t < t_1 + \alpha(v)$ and $\ENV(v,t) \neq 0$}{\Reject}
      \lIf{$t \ge t_1 + \beta(v)$ and $\ENV(v,t) \neq 1$}{\Reject}
    }
  }
  \Accept\;
  \caption{Algorithm for the case of general $T$ based on network
  decompositions}
  \label{alg:low-diam}
\end{algorithm}

\paragraph{Approach.}
Let us recall the relevant details of the strategy of
\textcite{nakar21_back_icalp}.
In their paper, \citeauthor{nakar21_back_icalp} studied local rules resembling
the majority rule in the restricted setting where $G$ is a path.
Their idea involved splitting the path into intervals that intersect at periodic
control points.
The first queries made obtain the state of these control points at a certain
time step $t_1$.
If there is no initial configuration leading to what we are observing at $t_1$
(i.e., the configuration we are observing is not \emph{feasible}), then we can
immediately reject.
Otherwise we can use the states of the nodes at the control points (plus some
additional queries) to fully predict almost the entirety of $\ENV$ after $t_1$.
Hence we only need to test a certain number of times if $\ENV(v,t)$ matches our
prediction where $(v,t) \in V \times \{ t \in \N_+ \mid t \ge t_1 \}$ is chosen
uniformly at random.

Our approach is more or less the same, though we need to cater for a couple
differences between our setting and theirs.
We are not in a path, and so in general we cannot split our graph into intervals
of the same size; rather we must work with a graph decomposition, which does
give us the adequate control points (the vertices incident to the edges of the
cut $C$, which form the set $B$ in \cref{alg:low-diam}) but only an upper bound
on the diameter of each component (which correspond to the intervals in the
setting of \citeauthor{nakar21_back_icalp}).
Fortunately the $\oBP$ rule is much simpler than majority or the like, and hence
the prediction in each component is easier to make.
The relevant observation is that the $\oBP$ rule converges fast to an
(all-black) fixed point in graphs of small diameter. 
(Indeed, the $\oBP$ rule converges in at most $\diam(G)$ steps.)
More specifically, components that started in an all-zero configuration must
stay zero until they enter in contact with a black node; meanwhile a component
$V_i$ that had at least one black node in it will converge to an all-black
configuration in at most $\diam(V_i) \le t_1$ steps.

Let us now give a more detailed overview of the steps performed by
\cref{alg:low-diam}.
For a set $S \subseteq V$ and $t \in [T]$, we say that $\ENV(S,t)$ is
\emph{feasible} if there is $\ENV' \in \oBP$ such that $\ENV'(v,t) = \ENV(v,t)$
for every $v \in S$.
\cref{thm:alg-low-diam} first sets $t_1$ appropriately and determines a graph
decomposition of $G$ where the components $V_1,\dots,V_r$ have diameter at most
$t_1$.
We wait for $t_1$ steps to elapse and then query the states of $B$, which are
the nodes incident to the edge cut $C$ of the graph decomposition, and can
immediately reject if what we see is not feasible.
At the same time we query a uniformly sampled set $Q$ of pairs corresponding to
the states of nodes in time steps after $t_1$, whose values we shall use later.
We then set $B_i = B \cap V_i$ and $B_i'$ to the nodes that are black in $B_i$
in time step $t_1$.
With these we can then compute estimates $\alpha_i(v)$ and $\beta_i(v)$ for each
node $v$ and each component $V_i$.
These are only intended to be useful if $v$ is white in time step $t_1$ and are
determined as follows:
\begin{itemize}
  \item $\alpha_i(v)$ is a \emph{lower bound} on the number of time steps that
  elapse after $t_1$ until $v$ turns from white to black.
  To compute $\alpha_i(v)$, we consider both nodes in $B_i'$ (whose state in
  $t_1$ is known to us) and nodes in the inside of $V_i$ (whose state is unknown
  and which means we must assume that they are black).
  If there are no nodes in $B_i'$, then we know that $V_i$ was all white at the
  beginning and we just set $\alpha_i(v) = \infty$.
  \item $\beta_i(v)$ is an \emph{upper bound} on the number of time steps after
  $t_1$ until $v$ turns black at the latest.
  To compute $\beta_i(v)$ we take into account only nodes which we are sure that
  are black in $t_1$, that is, nodes in $B_i'$.
  Again, if $B_i'$ is empty, then $V_i$ must have been all white in the first
  time step; in that case we set $\beta_i(v) = \infty$.
\end{itemize}
See \cref{fig:alg-low-diam} for an example.
Based on these estimates, we can then use the values of $Q$ to make random tests
on the state of nodes after $t_1$ based on what we know from $B_i'$ and
$\alpha(v) = \min_i \alpha_i(v)$ and $\beta(v) = \min_i \beta_i(v)$.
More precisely, for a pair $(v,t) \in Q$:
\begin{itemize}
  \item If $v$ was already black in time step $t_1$, then certainly it must
  still be black in time step $t \ge t_1$.
  \item Otherwise $v$ was white in time step $t_1$ and we can use our estimates
  $\alpha(v)$ and $\beta(v)$ to verify the predicted state of $v$ in step $t$,
  if possible.
\end{itemize}
The algorithm accepts by default if $\ENV$ passes the tests.

The query complexity of \cref{thm:alg-low-diam} is evident, so our main focus
now is on its correctness.

\begin{figure}
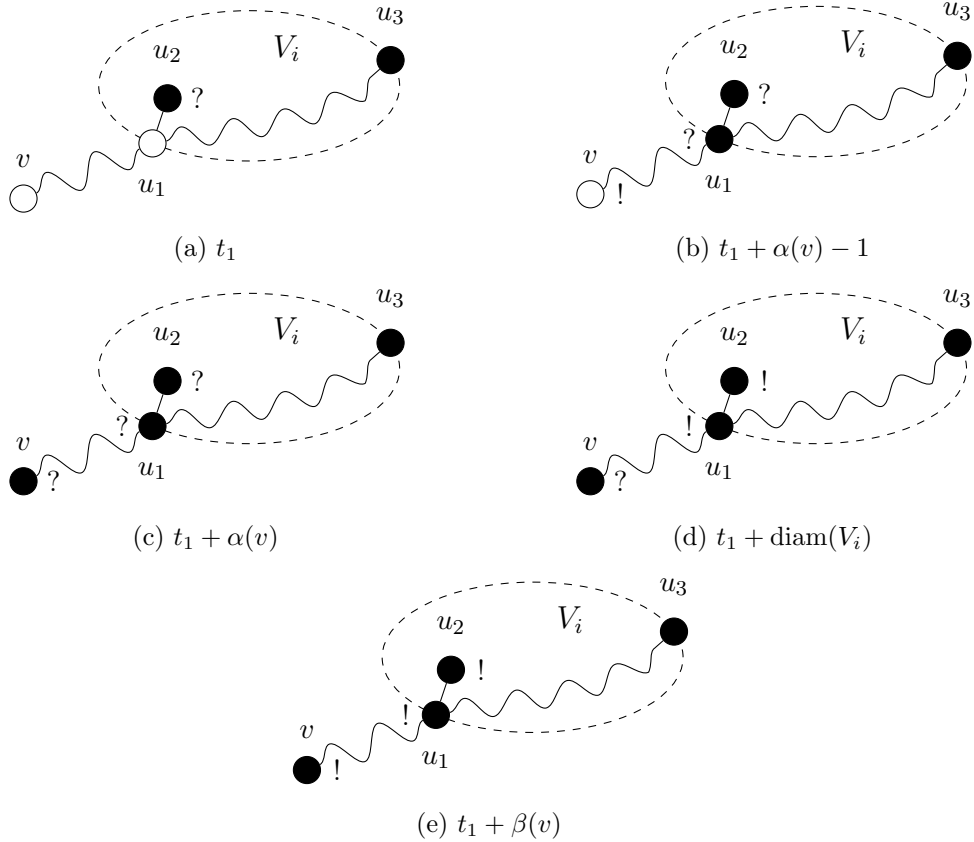

  \centering
  \hspace*{\fill}
  \subcaptionbox{$t_1$\label{fig:alg-low-diam-1}}{
    \includestandalone{figs/alg_low_diam_1}
  }
  \hspace*{\fill}
  \subcaptionbox{$t_1 + \alpha(v) - 1$\label{fig:alg-low-diam-2}}{
    \includestandalone{figs/alg_low_diam_2}
  }
  \hspace*{\fill}
  \\
  \hspace*{\fill}
  \subcaptionbox{$t_1 + \alpha(v)$\label{fig:alg-low-diam-3}}{
    \includestandalone{figs/alg_low_diam_3}
  }
  \hspace*{\fill}
  \subcaptionbox{$t_1 + \diam(V_i)$\label{fig:alg-low-diam-4}}{
    \includestandalone{figs/alg_low_diam_4}
  }
  \hspace*{\fill}
  \\
  \subcaptionbox{$t_1 + \beta(v)$\label{fig:alg-low-diam-5}}{
    \includestandalone{figs/alg_low_diam_5}
  }
  \caption{How to predict the color of a node $v$ based on knowledge about the
  states of nodes in other components.
  For the sake of illustration, here we are assuming that $v$ belongs to some
  component which is all white in step $t_1$ and that the component nearest to
  $v$ on which a black node appears is $V_i$.
  We also suppose that $B_i = \{ u_1, u_3 \}$ and $\dist(v,u_1) \ll
  \dist(u_1,u_3) = \diam(V_i)$.
  In time step $t_1$ the situation is as in (\subref{fig:alg-low-diam-1}).
  Since $B_i' = \{ u_3 \}$ is not empty, we must treat $V_i$ as potentially
  having black nodes since the first time step.
  We see the states of $u_1$ and $u_3$ in $t_1$ and determine that $\alpha(v) =
  \dist(v,u_1) + 1 = \dist(v,u_2)$ and $\beta(v) = \dist(v,u_3)$; however, we do
  not know the color of $u_2$ since it is inside $V_i$ and we do not query it,
  so we have to treat it as a potentially black node (denoted by a question
  mark).
  After $\alpha(v) - 1$ steps (\subref{fig:alg-low-diam-2}) we know that $v$
  must still be white (denoted by an exclamation mark) since the closest node to
  it that is possibly black in time step $t_1$ is the node $u_2$.
  After $\alpha(v)$ steps (\subref{fig:alg-low-diam-3}) we are no longer certain
  about the color of $v$.
  After $\diam(V_1)$ steps (\subref{fig:alg-low-diam-4}) we know that $u_1$ must
  be black, but we still cannot say anything about $v$.
  Finally after $\beta(v)$ steps (\subref{fig:alg-low-diam-5}) we are sure that
  $v$ has turned black at the latest since $u_3$ was black in $t_1$.}
  \label{fig:alg-low-diam}
\end{figure}

\paragraph{Correctness.}
The main idea for the correctness proof is that there is a partition $K + U$
(\enquote{known} and \enquote{unknown}) of $V \times [T]$ with the following
properties:
\begin{itemize}
  \item Given the values of $\ENV(B,t_1)$, we can \enquote{predict} the values
  of $k \in K$ if the $\oBP$ rule is followed correctly; that is, there is
  $p\colon K \to \binalph$ such that, if $\ENV \in \oBP$, then $\ENV(k) = p(k)$
  for every $k \in K$. 
  (Conversely, if there is $k \in K$ with $\ENV(k) \neq p(k)$, then $\ENV \notin
  \oBP$.)
  \item $\abs{U} \le \eps nT/2$, that is, $U$ is small compared to $K$.
\end{itemize}
Hence, given $\ENV(B,t_1)$, we can afford to ignore $U$ and need only perform an
additional $O(1/\eps)$ independent queries of $K$ in order to test $\ENV$.

We now focus on showing the existence of $K$.
Observe that, if $\ENV \in \oBP$, then the following holds for every $i \in
[r]$:
\begin{itemize}
  \item If there is $b \in B_i$ such that $\ENV(b,t_1) = 1$, then $\ENV(v,t) =
  1$ for every $v \in V_i$ and $t \ge t_1 + \diam(V_i)$ (since it takes at most
  $\diam(V_i) \le t_1$ steps for black to spread across $V_i$).
  \item If $\ENV(b,t_1) = 0$ for every $b \in B_i$, then necessarily
  $\ENV(v,t_1) = 0$ for every $v \in V_i$ (since otherwise we would have
  $\ENV(b,t_1) = 1$ for at least one $b \in B_i$).
\end{itemize}
Hence we can add pairs to $K$ and set $p(k)$ for $k \in K$ as follows:
\begin{itemize}
  \item Firstly, if $\ENV(v,t_1) = 1$ for a node $v$, then we can add $(v,t)$ to
  $K$ and set $p(v,t) = 1$ for every $t \ge t_1$.
  \item Suppose that $\ENV(v,t_1) = 0$ for some node $v$.
  For every $i \in [r]$, assuming $\ENV \in \oBP$ we have that $\alpha_i(v)$ is
  the length of the shortest path in time step $t_1$ between $v$ to any node
  that could \enquote{possibly} be black in $V_i$; that is, we consider both
  distances from $v$ to nodes in $B_i'$ (which we know for sure that are black)
  and to nodes in $V_i \setminus B_i$ (which could be black, but we cannot say
  for sure).
  If we know that $V_i$ is all-white in step $t_1$ (since all of $B_i$ is
  white), then we set $\alpha_i(v) = \infty$ as $v$ will certainly not turn
  black on account of a node in $V_i$.
  With these observations we can add $(v,t)$ to $K$ and set $p(v,t) = 0$ for
  every $t_1 \le t < t_1 + \alpha(v)$.
  \item Again suppose that $\ENV(v,t_1) = 0$ for a node $v$.
  For every $i \in [r]$, assuming $\ENV \in \oBP$ we have that $\beta_i(v)$ is
  the minimum distance in time step $t_1$ between $v$ and a node in $B_i'$
  (which we know for sure is black).
  Hence we add $(v,t)$ to $K$ for $t \ge t_1 + \beta(v)$ and also set $p(v,t) =
  1$.
\end{itemize}

We observe the definition of $\alpha_i(v)$ and $\beta_i(v)$ obeys the following:

\begin{claim}%
  \label{claim:alg-low-diam-alpha-beta}
  For every $v \in V$ and every $i \in [r]$, $\alpha_i(v) \le \beta_i(v)$ and
  $\beta_i(v) - \alpha_i(v) \le \diam(V_i)$.
\end{claim}
\begin{proof}
  The first inequality is evident since in $\alpha_i(v)$ we are taking the
  minimum over a larger set of nodes than in $\beta_i(v)$.
  For the second inequality notice first that, if $v \in V_i$, then certainly
  the inequality holds as $\beta_i(v) \le \diam(V_i)$.
  Hence let us assume that $v \notin V_i$.
  Let $b \in V_i$ be such that $\dist(b,v)$ is minimized.
  Notice that $b \in B_i$ since $v$ is outside $V_i$.
  If there is one such $b$ with $\ENV(b,t_1) = 1$, then $\alpha_i(v) =
  \beta_i(v)$, so suppose additionally that $\ENV(b,t_1) = 0$ for every such
  $b$.
  Then in the worst case we have that any node $u \in V_i$ with $\dist(u,v) =
  \alpha_i(v)$ must be one hop further from $v$ than $b$, and so $\alpha_i(v)
  \ge \dist(b,v) + 1$.
  Meanwhile any $u \in B_i'$ with $\dist(u,v) = \beta_i(v)$ is certainly at most
  $\diam(V_i)$ hops away from $b$, and so $\beta_i(v) \le \dist(b,v) +
  \diam(V_i)$.
\end{proof}

With the previous observations we have that \cref{alg:low-diam} always accepts
$\ENV \in \oBP$.
Hence all that remains is to show the following:

\begin{claim}
  If \cref{alg:low-diam} accepts $\ENV$ with at least $1/2$ probability, then
  $\dist(\ENV,\oBP) < \eps$.
\end{claim}

\begin{proof}
  Since $A$ rejects if $\ENV(B,t_1)$ is not feasible, there is some $\ENV' \in
  \oBP$ with $\ENV(B,t_1) = \ENV'(B,t_1)$.
  From \cref{claim:alg-low-diam-alpha-beta} we have that $\abs{U} \le \eps nT/2$
  since, for every $v \in V$, all but at most $\eps T/2$ pairs $(v,t)$ are in
  $K$ (since $(v,t) \notin K$ if and only if $t < t_1$ or $t_1 + \alpha(v) \le t
  < t_1 + \beta(v)$).
  On the other hand $\abs{Q} \ge 4/\eps$ implies that at most an $\eps/4$
  fraction of the pairs in $K$ must be such that $\ENV(k)$ agrees with $p(k)$,
  and certainly also $\ENV'(k) = p(k)$.
  It follows that at most a $3\eps/4 < \eps$ fraction of the pairs disagree
  between $\ENV$ and $\ENV'$.
\end{proof}

This concludes the proof of \cref{thm:alg-low-diam}.


\section*{Acknowledgments}

Augusto Modanese is supported by the Helsinki Institute for Information
Technology (HIIT).
Parts of this work were done while Augusto Modanese was affiliated with the
Karlsruhe Institute of Technology (KIT) and visiting the NII in Tokyo, Japan as
an International Research Fellow of the Japan Society for the Promotion of
Science (JSPS).
Yuichi Yoshida is partly supported by JSPS KAKENHI Grant Number 18H05291 and
20H05965.

We would like to thank Jukka Suomela for interesting discussions.


\printbibliography

\end{document}

%% file: figs/results_t_2.tex
\begin{tikzpicture}
  \draw[ultra thin,color=gray] (-0.1,-0.1) grid (4.1,4.1);

  \draw[->] (-0.2,0) -- (4.3,0) node[right] {$\Delta$};
  \draw[->] (0,-0.2) -- (0,4.3) node[above] {$\#$ queries};

  \node[anchor=east] at (-0.1,1.3) {$n^{1/3}$};
  \node[anchor=east] at (-0.1,2) {$\sqrt{n}$};
  \node[anchor=east] at (-0.1,4) {$n$};

  \node[anchor=north] at (1.3,-0.1) {$n^{1/3}$};
  \node[anchor=north] at (2,-0.1) {$\sqrt{n}$};
  \node[anchor=north] at (4,-0.1) {$n$};

  \draw[ultra thin,color=gray,dashed] (-0.1,1.3) -- (1.3,1.3);
  \draw[ultra thin,color=gray,dashed] (1.3,-0.1) -- (1.3,1.3);

  \draw[very thick,color=darkgreen] (0,0.06) -- (4,4);
  \draw[very thick,color=blue] (0,0.02) -- (2,2) -- (4,2);
  \draw[very thick,color=orange] (0.02,0) -- (2,1.97) -- (4,0);
  \draw[very thick,color=red] (0.06,0) -- (1.3,1.24) parabola[bend at end] (4,0);

  \matrix [draw,column 2/.style={anchor=base west},column sep=.5ex,
    every node/.style={font=\footnotesize}] at (9,2) {
    \draw[fill=darkgreen] (-0.05,-0.05) rectangle ++(0.3,0.3);
    & \node {$1$-sided error, non-adaptive upper bound};\\
    \draw[fill=blue] (-0.05,-0.05) rectangle ++(0.3,0.3);
    & \node {$1$-sided error, adaptive upper bound};\\
    \draw[fill=orange] (-0.05,-0.05) rectangle ++(0.3,0.3);
    & \node {$1$-sided error, adaptive lower bound};\\
    \draw[fill=red] (-0.05,-0.05) rectangle ++(0.3,0.3);
    & \node {$2$-sided error, non-adaptive lower bound};\\
  };
\end{tikzpicture}

%% file: figs/violation_type_I.tex
\begin{tikzpicture}[every node/.style={circle,draw}]
  \node (x) [label=right:$v$] {};
  \node (y1) at (-0.8,1) {};
  \node (y4) at (0.8,1) {};
  \node[fill=black] (y2) at ($(y1)!0.33!(y4)$) {};
  \node (y3) at ($(y1)!0.66!(y4)$) {};
  \draw (x) -- (y1)  (x) -- (y2)  (x) -- (y3)  (x) -- (y4);

  \node[draw=none] at ($(x) - (2.5,0)$) (stept) {Step $t$};
  \node[draw=none] at ($(stept) + (0,1)$) {Step $t-1$};

  \node[draw=none] at ($(y4) + (1,0)$) {$N(v)$};
  \draw[dashed] ($(x) + (0,1)$) ellipse (1.2 and 0.4);
\end{tikzpicture}

%% file: figs/violation_type_II.tex
\begin{tikzpicture}[every node/.style={circle,draw}]
  \node[fill=black] (x) [label=right:$v$] {};
  \node (y1) at (-0.8,1) {};
  \node (y4) at (0.8,1) {};
  \node (y2) at ($(y1)!0.33!(y4)$) {};
  \node (y3) at ($(y1)!0.66!(y4)$) {};
  \draw (x) -- (y1)  (x) -- (y2)  (x) -- (y3)  (x) -- (y4);

  \node[draw=none] at ($(x) - (2.5,0)$) (stept) {Step $t$};
  \node[draw=none] at ($(stept) + (0,1)$) {Step $t-1$};

  \node[draw=none] at ($(y4) + (1,0)$) {$N(v)$};
  \draw[dashed] ($(x) + (0,1)$) ellipse (1.2 and 0.4);
\end{tikzpicture}

%% file: figs/sets_alg_large_degree.tex
\begin{tikzpicture}
  \draw (0,0) grid (3,1);
  \draw[yshift=-0.5cm] (0,2) grid (3,3);

  \begin{scope}[shift={(0.5,0.5)}]
    \node (X2) at (0,0)  {$X_2$};
    \node (Y2) at (1,0) {$Y_2$};
    \node (Z2) at (2,0) {$Z_2$};
    \draw (-0.35,0.35) circle (1.5pt);
    \draw[fill] (0.65,0.35) circle (1.5pt) node[xshift=1.2mm] {\scriptsize !};
    \draw[fill] (1.65,0.35) circle (1.5pt);
  \end{scope}
  \begin{scope}[shift={(0.5,0)}]
    \node (X1) at (0,2) {$X_1$};
    \node (Y1) at (1,2) {$Y_1$};
    \node (Z1) at (2,2) {$Z_1$};
    \draw (-0.35,2.35) circle (1.5pt) node[xshift=1.2mm] {\scriptsize !};
    \draw[fill] (0.65,2.35) circle (1.5pt);
    \draw (1.65,2.35) circle (1.5pt);
  \end{scope}

  \node (W1) at ($(X1)+(0.17,0.9)$) {$W_1$};
  \draw[dashed] (-0.1,1.4) rectangle (W1.north east);

  \node (B2) at ($(Y2)+(0.2,-0.9)$) {$B_2$};
  \draw[dashed] (0.9,1.1) rectangle (B2.south east);
\end{tikzpicture}

%% file: figs/alg_low_diam_1.tex
\begin{tikzpicture}[every node/.style={draw,circle}]
  \node (v) at (0,0) [label=$v$] {};
  \coordinate (vi) at (3,1.5);
  \draw[dashed] (vi) ellipse (2cm and 1cm);
  \node[draw=none] at ($(vi) + (0.5,0.5)$) {\large $V_i$};
  \node[fill=white] (u1) at ($(vi) + (230:2cm and 1cm)$)
    [label=below:$u_1$] {};
  \node[fill=black] (u2) at ($(u1) + (0.2,0.6)$) [label=$u_2$] {};
  \node[draw=none] at ($(u2) + (0.4,0)$) {?};
  \node[fill=black] (u3) at ($(vi) + (20:2cm and 1cm)$) [label=$u_3$] {};
  \draw decorate [decoration={snake,segment length=7mm,amplitude=2mm}] 
    {(v) -- (u1)};
  \draw (u1) -- (u2);
  \draw decorate [decoration={snake,segment length=7mm,amplitude=1.5mm}]
    {(u1) to [out=10,in=220] (u3)};
\end{tikzpicture}

%% file: figs/alg_low_diam_2.tex
\begin{tikzpicture}[every node/.style={draw,circle}]
  \node (v) at (0,0) [label=$v$] {};
  \node[draw=none] at ($(v) + (0.4,0)$) {!};
  \coordinate (vi) at (3,1.5);
  \draw[dashed] (vi) ellipse (2cm and 1cm);
  \node[draw=none] at ($(vi) + (0.5,0.5)$) {\large $V_i$};
  \node[fill=black] (u1) at ($(vi) + (230:2cm and 1cm)$)
    [label=below:$u_1$] {};
  \node[draw=none] at ($(u1) - (0.4,0)$) {?};
  \node[fill=black] (u2) at ($(u1) + (0.2,0.6)$) [label=$u_2$] {};
  \node[draw=none] at ($(u2) + (0.4,0)$) {?};
  \node[fill=black] (u3) at ($(vi) + (20:2cm and 1cm)$) [label=$u_3$] {};
  \draw decorate [decoration={snake,segment length=7mm,amplitude=2mm}] 
    {(v) -- (u1)};
  \draw (u1) -- (u2);
  \draw decorate [decoration={snake,segment length=7mm,amplitude=1.5mm}]
    {(u1) to [out=10,in=220] (u3)};
\end{tikzpicture}

%% file: figs/alg_low_diam_3.tex
\begin{tikzpicture}[every node/.style={draw,circle}]
  \node[fill=black] (v) at (0,0) [label=$v$] {};
  \node[draw=none] at ($(v) + (0.4,0)$) {?};
  \coordinate (vi) at (3,1.5);
  \draw[dashed] (vi) ellipse (2cm and 1cm);
  \node[draw=none] at ($(vi) + (0.5,0.5)$) {\large $V_i$};
  \node[fill=black] (u1) at ($(vi) + (230:2cm and 1cm)$)
    [label=below:$u_1$] {};
  \node[draw=none] at ($(u1) - (0.4,0)$) {?};
  \node[fill=black] (u2) at ($(u1) + (0.2,0.6)$) [label=$u_2$] {};
  \node[draw=none] at ($(u2) + (0.4,0)$) {?};
  \node[fill=black] (u3) at ($(vi) + (20:2cm and 1cm)$) [label=$u_3$] {};
  \draw decorate [decoration={snake,segment length=7mm,amplitude=2mm}] 
    {(v) -- (u1)};
  \draw (u1) -- (u2);
  \draw decorate [decoration={snake,segment length=7mm,amplitude=1.5mm}]
    {(u1) to [out=10,in=220] (u3)};
\end{tikzpicture}

%% file: figs/alg_low_diam_4.tex
\begin{tikzpicture}[every node/.style={draw,circle}]
  \node[fill=black] (v) at (0,0) [label=$v$] {};
  \node[draw=none] at ($(v) + (0.4,0)$) {?};
  \coordinate (vi) at (3,1.5);
  \draw[dashed] (vi) ellipse (2cm and 1cm);
  \node[draw=none] at ($(vi) + (0.5,0.5)$) {\large $V_i$};
  \node[fill=black] (u1) at ($(vi) + (230:2cm and 1cm)$)
    [label=below:$u_1$] {};
  \node[draw=none] at ($(u1) - (0.4,0)$) {!};
  \node[fill=black] (u2) at ($(u1) + (0.2,0.6)$) [label=$u_2$] {};
  \node[draw=none] at ($(u2) + (0.4,0)$) {!};
  \node[fill=black] (u3) at ($(vi) + (20:2cm and 1cm)$) [label=$u_3$] {};
  \draw decorate [decoration={snake,segment length=7mm,amplitude=2mm}] 
    {(v) -- (u1)};
  \draw (u1) -- (u2);
  \draw decorate [decoration={snake,segment length=7mm,amplitude=1.5mm}]
    {(u1) to [out=10,in=220] (u3)};
\end{tikzpicture}

%% file: figs/alg_low_diam_5.tex
\begin{tikzpicture}[every node/.style={draw,circle}]
  \node[fill=black] (v) at (0,0) [label=$v$] {};
  \node[draw=none] at ($(v) + (0.4,0)$) {!};
  \coordinate (vi) at (3,1.5);
  \draw[dashed] (vi) ellipse (2cm and 1cm);
  \node[draw=none] at ($(vi) + (0.5,0.5)$) {\large $V_i$};
  \node[fill=black] (u1) at ($(vi) + (230:2cm and 1cm)$)
    [label=below:$u_1$] {};
  \node[draw=none] at ($(u1) - (0.4,0)$) {!};
  \node[fill=black] (u2) at ($(u1) + (0.2,0.6)$) [label=$u_2$] {};
  \node[draw=none] at ($(u2) + (0.4,0)$) {!};
  \node[fill=black] (u3) at ($(vi) + (20:2cm and 1cm)$) [label=$u_3$] {};
  \draw decorate [decoration={snake,segment length=7mm,amplitude=2mm}] 
    {(v) -- (u1)};
  \draw (u1) -- (u2);
  \draw decorate [decoration={snake,segment length=7mm,amplitude=1.5mm}]
    {(u1) to [out=10,in=220] (u3)};
\end{tikzpicture}